\documentclass[12pt]{article}
\usepackage[utf8]{inputenc}
\usepackage[english]{babel}
\usepackage[margin=1in]{geometry}
\usepackage[format=hang, font={small, it}, labelfont=bf]{caption}
\usepackage{multirow}

\usepackage{amsmath,amssymb,amsthm}
\usepackage{mathrsfs}
\usepackage{xfrac}
\usepackage[overload]{empheq}

\usepackage{graphicx}
\usepackage{threeparttable}
\usepackage{hyperref}
\usepackage{tabularx}
\usepackage{booktabs}
\usepackage{ragged2e}

\usepackage{lmodern}
\usepackage{thmtools}
\usepackage{enumitem}
\usepackage{indentfirst}

\usepackage{lscape}
\usepackage{tikz, tikz-cd, pgfplots}
\pgfplotsset{compat=1.18}
\usepackage{physics}
\usepackage{array} 
\usepackage{tabularx}

\usepackage[affil-it]{authblk}

\hypersetup{hidelinks}

\DeclareMathOperator{\cn}{cn}
\DeclareMathOperator{\sn}{sn}
\DeclareMathOperator{\dn}{dn}

\DeclareMathOperator{\K}{K}

\newtheorem{definition}{Definition}[section]
\newtheorem{theorem}{Theorem}[section]
\newtheorem{lemma}{Lemma}[section]

\numberwithin{equation}{section}
\allowdisplaybreaks[4]

\usepackage[square,numbers, sort&compress]{natbib}
\title{
Dynamics and Statistics of Dissipative Nonlinear Systems: A Study via 2D CGLE by Contact Geometry}
\author[1,2]{D.Y. Zhong}
\author[2]{G. Q. Wang}
\affil[1]{State Key Laboratory of Hydroscience and Engineering, Tsinghua University, Beijing 100084, China }
\affil[2]{Department of Hydraulic Engineering, Tsinghua University, Beijing 100084, China }
\date{}

\begin{document}

\maketitle

\begin{abstract}

A unified description of dynamics and statistics is fundamental to understanding dissipative nonlinear systems. This paper develops one using the contact geometry of complex stochastic vector bundles. The first theorem establishes the least-constraint principle for complex vector bundles. It uniquely determines the characteristic field, whose integral curves obey the contact Hamilton equations and preserve the contact form exactly, leading to the dissipative contact Hamilton-Jacobi equation. The second theorem constructs a probability measure from the same contact structure.  Applied to the two-dimensional cubic--quintic complex Ginzburg--Landau equation, the contact Hamilton--Jacobi equation admits closed-form travelling waves in reciprocal Jacobi-elliptic form. The induced measure yields a heavy-tailed intensity law $I^{-3/2}$ with non-extensive index $q=5/3$, anomalous structure-function scaling, an on-shell bifurcation between extended waves and localised solitons with square-root onset, and an equilibrium crossover converging in the thermodynamic limit.


\end{abstract}

\tableofcontents
\listoffigures
\listoftheorems
\listoftables

\newpage

\section{Introduction}\label{sec:introduction}

Dissipative nonlinear dynamics describes nonequilibrium systems that sustain spatiotemporal structures through a balance of external driving, energy dissipation, and nonlinear interactions \cite{Mori1998DissipativeStructures,Cross1993PatternFormation}. Such structures occur across scales, from convective patterns in fluids \cite{Schoepf1993Convection,Khater2006SolitonSolutions} and coherent light in nonlinear optical systems \cite{Malomed2021FractionalMedia} to defect-mediated textures in condensed matter \cite{Teo2017TopologicalDefects}. A central problem is to describe both the structures selected by the dynamics and their statistical properties.

The Complex Ginzburg--Landau Equation (CGLE) is a paradigmatic model of dissipative pattern formation \cite{Aranson2002CGLEWorld}. Its basic form balances driving forces, such as thermal buoyancy in convection and gain in optical systems, against dissipation and captures the transition from homogeneous steady states to ordered periodic patterns \cite{Kengne2021ModulationalInstability,Scheel2003RotatingConvection,Kumar2025SolitonStructures}. Applications include superfluid neutron-star matter \cite{Richardson1972NeutronStar}, nonlinear transmission networks \cite{Kengne2022TransmissionNetworks}, birefringent optical fibres \cite{Yomba2024CoupledSolitons}, and dissipative localised waves \cite{Tsoy2006LocalizedWaves,Descalzi2025DissipativePulses}. Cubic--quintic and higher-order variants describe modulation instability in nonlinear metamaterials \cite{TiamMegne2020ModulationInstability}, few-cycle light bullets \cite{Tabi2024LightBullets}, optical vortex and necklace-ring beams \cite{Djoko2019OpticalVortices,TiamMegne2023VortexBeams}, septic nonlinearity \cite{Djoko2021SepticNonlinearity}, spin-orbit coupled vortex rings in exciton-polariton condensates \cite{Tabi2025VortexRings}, and time-varying anisotropic media \cite{VanGorder2024AnisotropicMedia}. The two-dimensional cubic--quintic form studied here can arise as the amplitude equation of a subcritical instability, with cubic amplification balanced by quintic saturation \cite{Cross1993PatternFormation,Aranson2002CGLEWorld}.

The dynamics of the CGLE, including its solutions, their stability and their bifurcations, has been studied extensively \cite{Akhmediev1997Multisoliton,Rezazadeh2018KerrSolitons,Akhmediev2001PulsatingSolitons,Akhmediev1995ArbitraryAmplitude,Wang2025FractionalCGLE} through analytical techniques and numerical methods \cite{Montagne1996LyapunovFunctional,Mancas2006Wavetrains}. Its statistical properties, by contrast, are obtained either from probability distributions measured in direct simulations \cite{Kim2003RareEvents,Iwasaki1992Turbulence} or from stationary measures of randomly forced systems \cite{Kuksin2004StationaryMeasures}. More generally, probability weights are specified by noise-induced path-integral or large-deviation formulations \cite{Onsager1953Fluctuations,Martin1973StatisticalDynamics,Freidlin1998RandomPerturbations}, by thermostats that sample prescribed ensembles \cite{Bravetti2016Thermostat}, or by invariant measures associated with geometric flows \cite{Bravetti2015CanonicalMeasure,Bravetti2020InvariantMeasures}. The dynamics and the statistical properties are therefore usually described within separate mathematical frameworks.

A large amount of theories and approaches have been developed to study dynamical systems. Among them, geometric methods have long been applied to spatiotemporal dynamics \cite{Abraham1978Mechanics,Arnold1989Mechanics,Arnold1989Contact}, and contact geometry provides a framework for dissipative structures \cite{Bravetti2017ContactHamiltonian,Zhong2024KineticEquation,Gaset2020ContactFormalisms,Geiges2008ContactTopology}. Contact Hamiltonian mechanics is a developed subject \cite{deLeon2019ContactSystems}: its variational construction goes back to Herglotz \cite{Herglotz1930ContactTransformations,Guenther1995HerglotzLectures}, with a Noether theorem \cite{Georgieva2003HerglotzVariational} and treatments of the principle in field theory \cite{Gaset2024HerglotzFieldTheories} and in geometric or time-dependent settings \cite{Massa2023HerglotzProblem,Cannarsa2020HerglotzLaxOleinik}. Contact Hamilton--Jacobi theory has been formulated through line bundles and their jets \cite{Grabowska2022ContactCHJ,deLeon2021ContactHJ}, with implicit and weak-KAM variants \cite{Esen2023ImplicitContact,Su2016WeakKAM}. Dissipative field theories have been treated through bracket formalisms \cite{Morrison1986HamiltonianDissipative,Grmela1984Bracket,Grmela1997ComplexFluids,Bloch1996EulerPoincare,Coquinot2020Metriplectic}, variational formulations of nonequilibrium thermodynamics \cite{GayBalmaz2017Discrete,GayBalmaz2017Continuum}, contact and multicontact geometry \cite{Gaset2020FieldTheories,deLeon2023Multicontact}, and metriplectic systems on contact manifolds \cite{Morrison2026MetriplecticContact}. For finite-dimensional contact flows, Liouville-type theorems yield invariant measures whose density with respect to the contact volume is a power of the contact Hamiltonian, $|h|^{-(n+1)}$ \cite{Bravetti2015CanonicalMeasure}, with existence results under completeness conditions \cite{Bravetti2020InvariantMeasures}. Related dissipative and thermodynamic constructions are developed in \cite{Ciaglia2018Dissipation,Grmela2014MesoscopicThermo}, while Ruppeiner geometry relates thermodynamic fluctuations to a metric \cite{Ruppeiner1995RiemannianThermo}. Recent work of the authors of this paper has applied contact geometry to stochastic dissipative systems: the stochastic vector-bundle formulation gives the probability space an infinite-order jet structure carrying a natural contact form \cite{Zhong2024KineticEquation}, with evolution fixed by a least-constraint theorem \cite{Zhong2025LeastConstraint}.

For dissipative systems in general, however, a probability measure over configurations that is determined by the same contact construction as the characteristic dynamics remains to be established. To address this problem, this paper constructs a complex stochastic vector bundle whose infinite-order probability jets carry a natural contact form \cite{Zhong2024KineticEquation,Zhong2025LeastConstraint}, and applies the resulting contact formulation to the two-dimensional (2D) cubic--quintic CGLE. This construction extends the stochastic vector-bundle formulation from real to complex fibres. The first theorem is the complex least-constraint theorem: the contact structure and the constraint determine a unique characteristic field, whose integral curves satisfy the contact Hamilton equations and preserve the contact form exactly. The second theorem is the probability-measure theorem: the coordinate volume determined by parallel transport and the action, together with characteristic projection and fixed-time on-shell disintegration, yields a probability density on the reduced configuration space. For the autonomous CGLE, the contact potential is conserved and plays the role of energy in the dissipative medium. After a canonical transformation and travelling-wave reduction, the associated dissipative contact Hamilton--Jacobi equation admits closed-form travelling waves in reciprocal Jacobi-elliptic form. This family's equilibrium weight exhibits an on-shell bifurcation between extended and localised patterns, and the induced measure yields a heavy-tailed intensity law $I^{-3/2}$ with non-extensive index $q=5/3$, anomalous structure-function scaling, and an equilibrium crossover whose centre and width converge in the thermodynamic limit, checked by two samplers of the constructed weight, an Euler--Maruyama radial diffusion and a Metropolis chain on the shape parameter. Therefore the closed-form travelling waves of the contact Hamilton--Jacobi equation and the statistical laws of the CGLE are downstream consequences of this single contact construction.

The paper is organised as follows. Section \ref{sec:contact-geometric-formulation} develops the complex contact structure and states the least-constraint theorem (Theorem~\ref{thm:least-constraint-complex}) and the probability measure (Theorem~\ref{thm:measure}). Section \ref{sec:Contact-HJE} derives the dissipative contact Hamilton--Jacobi (CHJ) equation for the 2D CGLE, and Section \ref{sec:Traveling-Wave-Solutions} reduces it to closed-form profiles in reciprocal Jacobi-elliptic form. Section \ref{sec:Contact-Probability-Functional} derives the reciprocal on-shell density (Theorem~\ref{thm:reciprocal-density}), identifies the reduced field as the image of the probability map, and analyses the statistics and the on-shell bifurcation between periodic and localised patterns; Section \ref{sec:Conclusions} summarises the results. The appendices contain the proofs of the main results together with the derivations of the travelling-wave, reciprocal-density, phase-diagram and simulation results.

\section{Contact-Geometric Formulation for Complex Fields}\label{sec:contact-geometric-formulation}

\subsection{Least-Constraint Theorem}
\label{sec:supplementary_theory}

We consider a contact structure defined by a contact 1-form $\Theta$ satisfying the non-degeneracy condition $\Theta\wedge(\dd\Theta)^{\wedge n}\neq0$ \cite{Arnold1989Contact}. Throughout, $n=\dim_{\mathbb{C}}E_{\mathbb{C}}$, and ranks and form degrees are counted over $\mathbb{C}$ unless stated otherwise.

The least-constraint theorem for real stochastic vector bundles with $E\subset\mathbb{R}^n$ was established in \cite{Zhong2025LeastConstraint}. Its extension to the complex amplitude fields $W\in\mathbb{C}^n$ considered here uses the following conventions.

The phase space has coordinates $(W^{j},\Pi^{*}_{j},t)$, with $W^{j}(\mathbf{r},t)$ the complex field as a section of the bundle. The symbol $W^{*}$ denotes the complex conjugate of $W$, and the fundamental brackets are those of the pairs $(W^{i},\Pi^{*}_{j})$. Formal derivatives with respect to $W$ and $W^{*}$ are taken with the other held fixed. For a real-valued function $F$ this convention implies
\begin{equation}
\frac{\partial F}{\partial W^{*}} = \overline{\left(\frac{\partial F}{\partial W}\right)} .
\label{eq:complex-derivative-convention}
\end{equation}

\begin{definition}[Stochastic Vector Bundle]\label{def:stochastic-vector-bundle}
Let $E_{\mathbb{C}}$ be an $n$-dimensional complex vector space and let $\mathbb{P}$ be the space of probability measures on $E_{\mathbb{C}}$. The stochastic vector bundle over $E_{\mathbb{C}}$, introduced in \cite{Zhong2024KineticEquation,Zhong2025LeastConstraint}, is the infinite-order jet bundle
\begin{equation}
\pi^{\infty}_{E_{\mathbb{C}},0}:\;J^{\infty}(E_{\mathbb{C}},\mathbb{P})\;\to\;E_{\mathbb{C}},
\label{eq:stochastic-jet-bundle}
\end{equation}
whose fibre over a point of $E_{\mathbb{C}}$ carries the jets of a probability measure $P\in\mathbb{P}$.
\end{definition}

\begin{definition}[Contact Structure]\label{def:contact-structure}
On the stochastic vector bundle of Definition~\ref{def:stochastic-vector-bundle}, similar to \cite{Zhong2025LeastConstraint}, the extended complex phase space is
\begin{equation}
\mathcal{E}_{\mathbb{C}}=T^{*}E_{\mathbb{C}}\times\mathbb{R},
\label{eq:extended-phase-space}
\end{equation}
with coordinates $(W^{j},\Pi^{*}_{j},t)$, where $W^{j}$ and $\Pi^{*}_{j}$ are complex and $t$ is real. Hence $\dim_{\mathbb{R}}\mathcal{E}_{\mathbb{C}}=4n+1$. The 1-form $\Theta_{\mathbb{C}}$ is complex-valued and $\ker\Theta_{\mathbb{C}}$ has complex rank $2n$. The conjugate variable to $W^{j}$ is
\begin{equation}
\Pi^{*}_{j}=\sum_{k=1}^{\infty}\ \sum_{\mu_1<\cdots<\mu_k}B^{\mu_1\cdots\mu_k}_{j}P_{\mu_1\cdots\mu_k},
\qquad
P_{\mu_1\cdots\mu_k}=\partial_{\mu_1\cdots\mu_k}P,\quad 1\le k\le\infty.
\label{eq:conjugate-field-complex}
\end{equation}
The coefficients $B^{\mu_1\cdots\mu_k}_{j}$ are assumed such that the map $P\mapsto\Pi^{*}$ is injective, equivalently such that $\dd P$, $\dd W^{j}$ and $\dd\Pi^{*}_{j}$ are linearly independent. The contact one-form is
\begin{equation}
\Theta_{\mathbb{C}}=\dd P-\Pi^{*}_{j}\,\dd W^{j}=\mathcal{H}\,\dd t-\Pi^{*}_{j}\,\dd W^{j},
\label{eq:contact-form-complex}
\end{equation}
with the contact probability potential
\begin{equation}
\mathcal{H}=\frac{\dd P}{\dd t}=\dd P(\dot{y})=\dot{y}\,\partial_{y}P,
\label{eq:contact-potential-def}
\end{equation}
and the complex contact form of degree $2n+1$
\begin{equation}
\Omega_{\mathbb{C}}=\Theta_{\mathbb{C}}\wedge(\dd\Theta_{\mathbb{C}})^{\wedge n},
\qquad
n=\frac{1}{2}\operatorname{rank}_{\mathbb{C}}\bigl(\dd\Theta_{\mathbb{C}}\vert_{\ker\Theta_{\mathbb{C}}}\bigr).
\label{eq:contact-volume-complex}
\end{equation}
The real contact volume is $\operatorname{Re}\Theta_{\mathbb{C}}\wedge(\dd\operatorname{Re}\Theta_{\mathbb{C}})^{\wedge 2n}$, of degree $4n+1$; $\Omega_{\mathbb{C}}$ is its complex counterpart.
\end{definition}

\begin{theorem}[Least-Constraint Theorem]\label{thm:least-constraint-complex}
In the setting of Definitions~\ref{def:stochastic-vector-bundle} and \ref{def:contact-structure}, and for a smooth complex function $\varepsilon_{\mathbb{C}}$ on $\mathcal{E}_{\mathbb{C}}$ with $\varepsilon_{\mathbb{C}}\neq0$, the following hold.
\begin{enumerate}[label=(\arabic*)]
    \item \textbf{Contact structure.} The form $\Omega_{\mathbb{C}}$ in
    Eq.~\eqref{eq:contact-volume-complex} is nowhere vanishing, so that $\Theta_{\mathbb{C}}$ is
    non-degenerate and characterises a non-trivial contact structure on
    $J^{\infty}(E_{\mathbb{C}},\mathbb{P})$.

    \item \textbf{Contact dynamics:} $\Theta_{\mathbb{C}}$ and
    $\varepsilon_{\mathbb{C}}$ determine a unique vector field
    \begin{equation}
    \mathfrak{X}(\mathcal{E}_{\mathbb{C}})=\{X_{\mathcal{H}}\mid X_{\mathcal{H}}\in\ker(\dd\Theta_{\mathbb{C}})\subset T\mathcal{E}_{\mathbb{C}}\}
    \label{eq:characteristic-set-complex}
    \end{equation}
    of the form
    \begin{equation}
    X_{\mathcal{H}}=\frac{\partial}{\partial t}+\dot{W}^{i}\frac{\partial}{\partial W^{i}}
    +\dot{\Pi}^{*}_{i}\frac{\partial}{\partial\Pi^{*}_{i}}.
    \label{eq:vector}
    \end{equation}
Let $\gamma\colon I\to\mathcal{E}_{\mathbb{C}}$ be an integral curve of
$X_{\mathcal{H}}$, that is $\dot\gamma(t)=X_{\mathcal{H}}\big|_{\gamma(t)}$,
with $I\subset\mathbb{R}$ an interval. If $\Theta_{\mathbb{C}}$ and $\varepsilon_{\mathbb{C}}=\varepsilon_{\mathbb{C}}(\gamma(t))$ satisfy  
    \begin{equation}
    \iota_{X_{\mathcal{H}}}\Theta_{\mathbb{C}}=-\varepsilon_{\mathbb{C}},
    \qquad
    \iota_{X_{\mathcal{H}}}\dd\Theta_{\mathbb{C}}=0,
    \qquad
    L_{X_{\mathcal{H}}}\Theta_{\mathbb{C}}=0,
    \label{eq:three-conditions}
    \end{equation}
the contact dynamics is
    \begin{subequations}\label{eq:contact-dynamics}
    \begin{empheq}[left=\empheqlbrace]{align}
    \frac{\partial W^{j}}{\partial t} &= \;\;\; \{W^{j}, \mathcal{H}\} =\;\; \frac{\partial  \mathcal{H}}{\partial  \Pi^*_{j}}, \\
    \frac{\partial \Pi^*_{j}}{\partial t} &= \;\;\;\{\Pi^*_{j}, \mathcal{H}\} = -\frac{\partial  \mathcal{H}}{\partial  W^{j}},\\
    \frac{\partial \varepsilon_{\mathbb{C}}}{\partial t} &  =\, -\{\varepsilon_{\mathbb{C}}, \mathcal{H}\}.
    \end{empheq}
    \end{subequations}
    with the bracket defined by
    \begin{equation}
    \{F, G\} = \frac{\partial  F}{\partial  W^{j}} \frac{\partial  G}{\partial  \Pi^*_{j}}
    - \frac{\partial  F}{\partial  \Pi^*_{j}} \frac{\partial  G}{\partial  W^{j}},
    \label{eq:poisson-bracket}
    \end{equation}
    for any two functions $F$ and $G$ of $(W,\Pi^{*})$.    The time evolution of any function $F$ is then
    \begin{equation}
    \frac{\dd F}{\dd t} = \{F, \mathcal{H}\} + \frac{\partial F}{\partial t}.
    \label{eq:dynamics-F}
    \end{equation}

    \item \textbf{Least constraint and the contact Hamilton--Jacobi equation:} The flow of
    $\mathfrak{X}(\mathcal{E}_{\mathbb{C}})$ corresponds to the extremal paths of the action
    functional of $\varepsilon_{\mathbb{C}}$, defined as
    \begin{equation}
    \mathcal{S}=\int\varepsilon_{\mathbb{C}}\,\dd t=-\int\Theta_{\mathbb{C}}.
    \label{eq:action-complex}
    \end{equation}
    The least
    constraint is the stationarity of the action under variations that vanish at the endpoints,
    \begin{equation}
    \delta\mathcal{S}=0 ,
    \label{eq:least-constraint}
    \end{equation}
    equivalently the contact Hamilton--Jacobi-type equation
    \begin{equation}
    \frac{\partial\mathcal{S}}{\partial t}+\mathcal{H}\!\left[W,\frac{\partial\mathcal{S}}{\partial W},t\right]=0,
    \label{eq:chj-complex}
    \end{equation}
    with $\Pi^{*}_{j}$ the derivative of $\mathcal{S}$ with respect to $W^{j}$.
\end{enumerate}
\end{theorem}

\begin{proof}

The proof is given in Appendix~\ref{app:proof-thm-complex}.
\end{proof}
\subsection{Probability Measure}

The contact form $\Theta_{\mathbb{C}}$ also determines a complex line bundle and a connection on it. Parallel transport for this connection, taken together with the flow connection along an integral curve of $X_{\mathcal{H}}$, is what lets the action $\mathcal{S}$ of Eq.~\eqref{eq:action-complex} act on the linear algebra of $\ker\Theta_{\mathbb{C}}$.

\begin{lemma}[Contact Connection]\label{lem:contact-connection}
On the contact manifold $(\mathcal{E}_{\mathbb{C}},\Theta_{\mathbb{C}})$ of Definition~\ref{def:contact-structure}, let
\begin{equation}
Q:=T\mathcal{E}_{\mathbb{C}}\big/\ker\Theta_{\mathbb{C}}
\label{eq:contact-line-bundle}
\end{equation}
be the quotient line bundle. There is a connection $\nabla^{Q}$ on $Q$,
\begin{equation}
\nabla^{Q}_{X}\,\sigma=X(\sigma)-\Theta_{\mathbb{C}}(X)\,\sigma.
\label{eq:nablaQ-def}
\end{equation}
\end{lemma}

\begin{proof}
The proof is given in Appendix~\ref{app:proof-contact-connection}.
\end{proof}

The construction of the measure runs along an integral curve of $X_{\mathcal{H}}$ and uses the transport of Lemma~\ref{lem:contact-connection} at each fixed time. The second theorem of the paper is the resulting probability measure on the configuration space $E_{\mathbb{C}}$.

\begin{theorem}[Probability Measure]\label{thm:measure}
The contact structure $(\mathcal{E}_{\mathbb{C}},\Theta_{\mathbb{C}})$ of
Definition~\ref{def:contact-structure} determines on $E_{\mathbb{C}}$ the
probability measure
\begin{equation}
\mathcal{P}[W]=\frac{1}{\mathcal{Z}}\,e^{\,n\mathcal{S}[W]},
\qquad
\mathcal{Z}=\int_{E_{\mathbb{C}}}e^{\,n\mathcal{S}[W]}\,\mathcal{D}W ,
\label{eq:prob_measure}
\end{equation}
at each fixed time $t$, provided the
action $\mathcal{S}$ of Eq.~\eqref{eq:action-complex} is real-valued and
$0<\mathcal{Z}<\infty$; here $\mathcal{D}W$ is the coordinate volume element
of $E_{\mathbb{C}}$ and $n$ is the exponent of
Eq.~\eqref{eq:contact-volume-complex}.
\end{theorem}

\begin{proof}
The proof is given in Appendix~\ref{app:proof-measure} with the help of Lemma~\ref{lem:contact-connection}.
\end{proof}

\subsection{Relation to the Classical Formulation}
\label{subsec:comparison-contact}

The construction of this section differs from the standard contact Hamiltonian formulation in two respects.

The first concerns the origin of the measure. In the standard formulation \cite{Bravetti2017ContactHamiltonian,deLeon2019ContactSystems,Bravetti2015CanonicalMeasure}, a contact form $\eta$ and a Hamiltonian function $h$ specify the characteristic flow, and an invariant density is then obtained from an invariance condition. Here the contact potential is the rate of change of the probability measure along the state direction (Theorem~\ref{thm:least-constraint-complex}); Appendix~\ref{app:proof-measure}, Steps 1--3 determine the action-dependent coordinate volume, from which Theorem~\ref{thm:measure} gives the density.

The second concerns the three conditions Eq.~\eqref{eq:three-conditions}. In the standard formulation the contact form is preserved only conformally: with the Reeb field $\xi$ of $\eta$, Cartan's identity gives $\mathcal{L}_{X_h}\eta=\xi(h)\,\eta$ and $\mathcal{L}_{X_h}\bigl(\eta\wedge(d\eta)^{n}\bigr)=(n+1)\,\xi(h)\,\eta\wedge(d\eta)^{n}$, and the factor that renders the volume invariant is a power of the Hamiltonian, $\rho\propto |h|^{-(n+1)}$ \cite{Bravetti2015CanonicalMeasure}. In the present framework the contact form is preserved exactly by the third condition of Eq.~\eqref{eq:three-conditions}, since $X_{\mathcal{H}}$ lies in $\ker\dd\Theta_{\mathbb{C}}$ and $\dd\varepsilon_{\mathbb{C}}$ vanishes; the contact volume itself is therefore invariant. The corresponding density in the present construction is $\mathcal P \propto e^{n\mathcal{S}}$ (Theorem~\ref{thm:measure}).

\section{Geometric Dynamics of 2D CGLE}\label{sec:Contact-HJE}

\subsection{2D Cubic--Quintic CGLE}
\label{sec:Ginzburg-Landau-Equation}

The 2D cubic--quintic Complex Ginzburg--Landau Equation studied here is, in dimensionless form \cite{Aranson2002CGLEWorld},
\begin{equation}\label{eq:qcgle}
\frac{\partial W}{\partial t}=\mu W+(1+ic_{1})|W|^{2}W-(1+ic_{3})|W|^{4}W+(1+ic_{2})\nabla^{2}W,
\end{equation}
where $\nabla^{2}$ is the two-dimensional Laplacian in $x$ and $y$, and $\mu$, $c_{1}$, $c_{2}$ and $c_{3}$ are real constants.

Eq.~\eqref{eq:qcgle} is the amplitude equation of a subcritical instability: the cubic term amplifies and the quintic term saturates, and the sign structure together with the standard form are set out in \cite{Cross1993PatternFormation,Aranson2002CGLEWorld}. Any equation of this sign class is brought to Eq.~\eqref{eq:qcgle} by a rescaling of $W$, $t$ and the spatial coordinates, so $\mu$, $c_{1}$, $c_{2}$ and $c_{3}$ are its only free parameters.

\subsection{Contact Dynamics of CGLE}

\subsubsection{Dynamical Equations}
Equation~\eqref{eq:poisson-bracket} gives the bracket in $(W,\Pi^*)$; for a real intensity or a real action it is applied together with the complex-conjugate contribution, in accordance with Eq.~\eqref{eq:complex-derivative-convention}.

In the contact dynamics Eq.~\eqref{eq:contact-dynamics}, $\mathcal{H}$ can be determined by
\begin{equation}
\varepsilon_{\mathbb{C}} =  \Pi^*\dot{W}   - \mathcal{H}, \quad \dot{W} =\partial  \mathcal{H}/\partial  \Pi^*
\quad \Longrightarrow\quad \varepsilon_{\mathbb{C}} =  \Pi^*\frac{\partial  \mathcal{H}}{\partial  \Pi^*}   - \mathcal{H}.
\end{equation}
Following \cite{Zhong2025LeastConstraint}, we truncate the contact-potential expansion at linear order in $\Pi^*$ and include an additive real constant $C\neq0$:
\begin{equation}\label{eq:contact-potential}
\mathcal{H} = C + \Pi^* \dot{W}.
\end{equation}
The constant $C$ leaves the dynamics unchanged and keeps $\varepsilon_{\mathbb{C}}=-C\neq0$ along every integral curve (Eq.~\eqref{eq:eps-along-curve-complex}).

The Poisson bracket gives
\begin{align}
\dot{W}& = \frac{\partial  \mathcal{H}}{\partial  \Pi^*},
\end{align}
which is exactly Eq.~\eqref{eq:qcgle}. Similarly, 
\begin{align}
\dot\Pi^*&
=  -\frac{\partial  \mathcal{H}}{\partial  W}.
\end{align}

\subsubsection{Conservation and Intensity Balance}

\paragraph{Contact potential $\mathcal{H}$.}
Equation~\eqref{eq:contact-potential-def} defines $\mathcal{H}$ as the rate of change of the probability measure along the state direction. For any function $F$, the conservation criterion in contact dynamics is
\begin{equation}
\frac{\dd F}{\dd t} = \frac{\partial F}{\partial t} + \{F, \mathcal{H}\} = 0.
\label{eq:conservation-criterion}
\end{equation}
Taking $F=\mathcal{H}$ and using the antisymmetry $\{\mathcal{H},\mathcal{H}\}\equiv 0$ gives
\begin{equation}
\frac{\dd\mathcal{H}}{\dd t} = \{\mathcal{H}, \mathcal{H}\} + \frac{\partial \mathcal{H}}{\partial t} = \frac{\partial \mathcal{H}}{\partial t}.
\label{eq:H-conservation-condition}
\end{equation}
The contact potential is therefore conserved precisely when it carries no explicit time dependence, which holds for the autonomous CGLE studied here: $\dd\mathcal{H}/\dd t=0$.

\paragraph{Intensity: $\mathcal{N} = |W|^2 $.}
The Poisson bracket of $\mathcal{N}$ with the contact potential is
\begin{equation}
\begin{aligned}
\{\mathcal{N}, \mathcal{H}\}
&=   \frac{\partial  \mathcal{N}}{\partial  W} \frac{\partial  \mathcal{H}}{\partial  \Pi^*} - \frac{\partial  \mathcal{N}}{\partial  \Pi^*} \frac{\partial  \mathcal{H}}{\partial  W}   \\
&=   W^* \frac{\partial  \mathcal{H}}{\partial  \Pi^*} - 0 \cdot \frac{\partial  \mathcal{H}}{\partial  W} \\
&=  W^* \Bigl( \mu W + (1+ic_{1})|W|^{2}W - (1+ic_{3})|W|^{4}W + (1+ic_{2})\nabla^{2}W \Bigr) \\
&=  W^* \dot{W}.
\end{aligned}
\label{eq:N-bracket}
\end{equation}
Since $\mathcal{N}$ is real, its time derivative also contains the complex conjugate of Eq.~\eqref{eq:N-bracket},
\begin{equation}
\frac{\dd\mathcal{N}}{\dd t} = W^{*}\dot{W} + \dot{W}^{*}W = 2\,\mathrm{Re}\bigl(W^{*}\dot{W}\bigr) .
\label{eq:N-evolution}
\end{equation}

Substituting Eq.~\eqref{eq:qcgle} into Eq.~\eqref{eq:N-evolution} gives the explicit balance
\begin{equation}
\frac{1}{2}\,\frac{\dd\mathcal{N}}{\dd t}
=\mu\mathcal{N}+\mathcal{N}^{2}-\mathcal{N}^{3}
+\mathrm{Re}\bigl(W^{*}\nabla^{2}W\bigr)
-c_{2}\,\mathrm{Im}\bigl(W^{*}\nabla^{2}W\bigr).
\label{eq:N-balance}
\end{equation}

For a spatially uniform state both Laplacian terms vanish, and Eq.~\eqref{eq:N-balance} reduces to an algebraic balance in $\mathcal{N}$ alone,
\begin{equation}
\begin{gathered}
\frac{1}{2}\,\frac{\dd\mathcal{N}}{\dd t}=g(\mathcal{N})\equiv\mathcal{N}\bigl(\mu+\mathcal{N}-\mathcal{N}^{2}\bigr),\\
g(\mathcal{N})=0\iff\mathcal{N}=0\ \ \text{or}\ \ \mathcal{N}=\mathcal{N}_{\pm}\equiv\frac{1\pm\sqrt{1+4\mu}}{2}.
\end{gathered}
\label{eq:N-uniform}
\end{equation}
For $\mu>0$ the root $\mathcal{N}_{-}$ is negative, and the derivative of $g$ is positive at the zero state and negative at $\mathcal{N}_{+}$, where it takes the value $-\mathcal{N}_{+}-2\mu$. The zero state is therefore unstable, and $\mathcal{N}_{+}$ is the stable saturated intensity.

\subsubsection{Contact Hamilton--Jacobi Equation}
\label{subsec:canonical_transformation}
For the CGLE, the pointwise action $\mathcal{S}[W,t]$ is taken to be real-valued, as Theorem~\ref{thm:measure} requires, and defines the canonical conjugate field by
\begin{equation}\label{eq:complex_action_momentum}
\Pi^*(\vb{r},t) = \frac{\partial  \mathcal{S}[W,t]}{\partial  W}.
\end{equation}
Substituting this relation into the contact dynamical equation for the action Eq.~\eqref{eq:dynamics-F}, we get the CHJ equation for the 2D CGLE:
\begin{equation}
\frac{\partial  \mathcal{S}}{\partial  t} + \mathcal{H}\left[W, \frac{\partial  \mathcal{S}}{\partial  W}, t\right] = 0. \label{eq:hj_theorem}
\end{equation}

The direct solution of Eq.~\eqref{eq:hj_theorem} is difficult because the field $W(\vb{r},t)$ is complex-valued and depends on two spatial coordinates. We therefore apply a canonical transformation to a real field $\Phi(y,t)\in\mathbb{R}$ with real conjugate field $\Pi_\Phi(y,t)\in\mathbb{R}$.

We define a pointwise type-2 generating function $F_2[W, \Pi_\Phi, t]$ that depends on the original complex field $W(\vb{r},t)$ and the new real conjugate field $\Pi_\Phi(y,t)$:
\begin{equation}
F_2[W, \Pi_\Phi, t] =  \sqrt{\frac{2\pi}{J}} \exp\left(-i(k_x x + \theta(y,t))\right) W(\vb{r},t) \Pi_\Phi(y,t),
\label{eq:generating_functional}
\end{equation}
where $k_x \in \mathbb{R}$ is the fixed wavenumber in the $x$-direction; $\theta(y,t) \in \mathbb{R}$ is a prescribed real phase of the reduced ansatz; $J > 0$ is the total intensity $J = \iint |W|^2 \dd^2\vb{r}$, taken over an $x$-range of length $2\pi$ and the profile domain in $y$ specified below. The transformation is performed at fixed $J$, treated as a time-independent external parameter on the time window over which the statistical steady-state approximation $\partial J/\partial t=0$ holds.

The partial derivatives of $F_2$ determine the canonical transformation:
\begin{subequations}
\label{eq:transformation_relations}
\begin{empheq}[left=\empheqlbrace]{align}
\Pi^*(\vb{r},t) &= \frac{\partial  F_2}{\partial  W(\vb{r},t)}
= \sqrt{\frac{2\pi}{J}} \exp\left(-i(k_x x + \theta(y,t))\right) \Pi_\Phi(y,t),
\label{eq:momentum_transform}\\
\Phi(y,t) &= \frac{\partial  F_2}{\partial  \Pi_\Phi(y,t)}
= \sqrt{\frac{2\pi}{J}} \exp\left(-i(k_x x + \theta(y,t))\right) W(\vb{r},t).
\label{eq:field_transform}
\end{empheq}
\end{subequations}
Solving these relations explicitly gives the direct transformation:
\begin{subequations}
\label{eq:direct_transformation}
\begin{empheq}[left=\empheqlbrace]{align}
W(\vb{r},t) &= \sqrt{\frac{J}{2\pi}} \exp\left(i(k_x x + \theta(y,t))\right) \Phi(y,t),
\label{eq:W_to_Phi}\\
\Pi^*(\vb{r},t) &= \sqrt{\frac{2\pi}{J}} \exp\left(-i(k_x x + \theta(y,t))\right) \Pi_\Phi(y,t),
\label{eq:Pi_to_PiPhi}
\end{empheq}
\end{subequations}
The normalisation $\int_{-\infty}^{\infty} \Phi^2(y,t) \dd y = 1$, with Eq.~\eqref{eq:W_to_Phi} giving $|W|^{2}=\tfrac{J}{2\pi}\Phi^{2}$ independently of $x$ over the $2\pi$ range, yields the total intensity $J$.

The reduced variables are considered at each $y$, and Eq.~\eqref{eq:poisson-bracket} gives their fundamental bracket at that $y$,
\begin{equation}
\{\Phi,\Pi_\Phi\}
=
\frac{\partial \Phi}{\partial  W}
\frac{\partial \Pi_\Phi}{\partial  \Pi^{*}}
-
\frac{\partial \Phi}{\partial  \Pi^{*}}
\frac{\partial \Pi_\Phi}{\partial  W}.
\label{eq:PB_chain}
\end{equation}
From Eq.~\eqref{eq:field_transform} and Eq.~\eqref{eq:Pi_to_PiPhi} one finds
\begin{subequations}
\begin{empheq}[left=\empheqlbrace]{align}
\frac{\partial \Phi}{\partial  W}
&=
\sqrt{\frac{2\pi}{J}}\,\exp\!\bigl(-i(k_{x}x+\theta(y,t))\bigr),
\\
\frac{\partial \Pi_\Phi}{\partial  \Pi^{*}}
&=
\sqrt{\frac{J}{2\pi}}\,\exp\!\bigl(i(k_{x}x+\theta(y,t))\bigr),
\end{empheq}
\end{subequations}
while the cross terms $\partial \Phi/\partial \Pi^{*}$ and $\partial \Pi_\Phi/\partial  W$ are identically zero. The two phase factors are conjugate and cancel, so Eq.~\eqref{eq:PB_chain} gives
\begin{equation}
\{\Phi,\Pi_\Phi\}=1,
\qquad
\{\Phi,\Phi\}=\{\Pi_\Phi,\Pi_\Phi\}=0
\end{equation}
at every fixed $y$, so the fixed-intensity reduced map $(W,\Pi^{*})\mapsto(\Phi,\Pi_\Phi)$ is canonical.

The conjugate field $\Pi^*$ is given by Eq.~\eqref{eq:complex_action_momentum}. For the real-valued action, the chain rule contains both complex-conjugate branches (Eq.~\eqref{eq:complex-derivative-convention}):
\begin{equation}
\frac{\partial  \mathcal{S}}{\partial  \Phi(y,t)}
=  \frac{\partial  \mathcal{S}}{\partial  W} \frac{\partial  W}{\partial  \Phi(y,t)}
+ \frac{\partial  \mathcal{S}}{\partial  W^{*}} \frac{\partial  W^{*}}{\partial  \Phi(y,t)}.
\label{eq:chain_rule_action}
\end{equation}
From Eq.~\eqref{eq:W_to_Phi} and its complex conjugate, we have:
\begin{equation}
\frac{\partial  W}{\partial  \Phi(y,t)} = \sqrt{\frac{J}{2\pi}} \exp\left(i(k_x x + \theta(y,t))\right),
\qquad
\frac{\partial  W^{*}}{\partial  \Phi(y,t)} = \sqrt{\frac{J}{2\pi}} \exp\left(-i(k_x x + \theta(y,t))\right).
\end{equation}
Substitution into Eq.~\eqref{eq:chain_rule_action} gives
\begin{equation}
\begin{aligned}
\frac{\partial  \mathcal{S}}{\partial  \Phi(y,t)}
&= \sqrt{\frac{J}{2\pi}} \exp\left(i(k_x x + \theta(y,t))\right) \Pi^*(\vb{r},t) \\
&\quad + \sqrt{\frac{J}{2\pi}} \exp\left(-i(k_x x + \theta(y,t))\right) \overline{\Pi^*}(\vb{r},t) \\
&= \Pi_\Phi(y,t) + \Pi_\Phi(y,t) = 2\,\Pi_\Phi(y,t).
\end{aligned}
\end{equation}
Both terms equal $\Pi_\Phi(y,t)$ by Eq.~\eqref{eq:Pi_to_PiPhi} and its complex conjugate.

The field velocity prescribed by Eq.~\eqref{eq:qcgle} splits into an amplitude rate and a phase rate,
\begin{equation}
\sqrt{\frac{2\pi}{J}}\,\mathrm{e}^{-i(k_{x}x+\theta)}\,\dot W
=\dot\Phi+i\,\Phi\,\dot\psi,
\qquad
\Pi^{*}\dot W=\Pi_{\Phi}\bigl(\dot\Phi+i\,\Phi\,\dot\psi\bigr),
\label{eq:velocity-split}
\end{equation}
with $\psi=k_{x}x+\theta$ the total phase. The contact potential is
\begin{equation}
\mathcal{H}=C+\Pi^{*}\dot W=C+\Pi_{\Phi}\dot\Phi+i\,\Pi_{\Phi}\Phi\dot\psi .
\label{eq:contact-potential-complex}
\end{equation}
For the real action, $\partial_\Phi\mathcal{S}=2\Pi_\Phi$, so the CHJ equation contains the sum of the two complex-conjugate branches:
\begin{equation}
\mathcal{H}+\overline{\mathcal{H}}
=2C+\bigl(\Pi^{*}\dot W+\overline{\Pi^{*}}\dot W^{*}\bigr)
=2C+2\,\mathrm{Re}\bigl(\Pi^{*}\dot W\bigr)
=2\bigl(C+\Pi_{\Phi}\dot\Phi\bigr).
\label{eq:branch-sum}
\end{equation}
The common factor of two cancels when the profile equation is solved for $\Phi''$, but remains in the normalisation of $\mathcal{S}$.

Appendix~\ref{sec:Solution-CHJ} gives $\partial_y\theta=0$ for the profiles considered here, so the amplitude part of Eq.~\eqref{eq:qcgle} supplies $\dot{\Phi}$. Substituting it into Eq.~\eqref{eq:branch-sum} and using the CHJ equation~\eqref{eq:hj_theorem} gives the dissipative contact Hamilton--Jacobi (CHJ) equation
\begin{equation}
\label{eq:dissipative_hj}
\frac{\partial  \mathcal{S}}{\partial  t}
+2\left[C+\Pi_\Phi\left((\mu-k_{x}^{2})\,\Phi
+\frac{J}{2\pi}\,\Phi^{3}
-\frac{J^{2}}{4\pi^{2}}\,\Phi^{5}
+\Phi''\right)\right]
=0.
\end{equation}


\section{Solutions of the CHJ Equation}
\label{sec:Traveling-Wave-Solutions}

By Theorem~\ref{thm:measure}, the CGLE probability density is $\mathcal{P}\propto e^{n\mathcal{S}}$, so its statistical properties are determined entirely by the action. On a travelling-wave submanifold, the CHJ equation determines this action, from which the CGLE statistics follow without solving the CGLE directly.

\subsection{Travelling-Wave Reduction}
\label{subsec:HJ_traveling_wave}

On the travelling-wave submanifold $\mathcal{T}\subset\mathcal{E}_{\mathbb{C}}$, take $\theta=-\omega t+\varphi(y)$ and $\Phi(y,t)\to\Phi(y)$, so that $\psi=k_x x-\omega t+\varphi(y)$. A localised or periodic profile forces $\varphi'=\kappa_0=0$ (Appendix~\ref{sec:Solution-CHJ}, Step~0), so $\varphi$ is constant and can be absorbed into the reference phase; the ansatz is
\begin{subequations}
\begin{empheq}[left=\empheqlbrace]{align}
W_{e}(x,y,t) &= \sqrt{\frac{J}{2\pi}}\,\exp\bigl(i(k_x x - \omega t)\bigr)\,\Phi(y), \label{eq:tw-ansatz} \\
\Pi_e^*(x,y,t) &= \sqrt{\frac{2\pi}{J}}\,\exp\bigl(-i(k_x x - \omega t)\bigr)\,\Pi_\Phi(y),
\end{empheq}
\end{subequations}
where the subscript $e$ labels the on-shell variables of the reduced profile.

On $\mathcal{T}$, take the branch-summed action in the form
\begin{equation}
\mathcal{S}(y,\Phi,\Pi_\Phi,t) = -2\,\omega t - 2\ln\Phi(y) + \ln\!\left(\frac{J}{2\pi}\right).
\end{equation}
The derivatives are:
\begin{equation}
\frac{\partial \mathcal{S}}{\partial t} = -2\,\omega, \quad 
\frac{\partial \mathcal{S}}{\partial \Phi(y)} = -\frac{2}{\Phi(y)} = 2\,\Pi_\Phi(y),
\end{equation}
so $\Pi_\Phi=-1/\Phi$. Substitution into Eq.~\eqref{eq:dissipative_hj} and solving for $\Phi''$ gives the second-order profile equation
\begin{equation}\label{eq:reduced-ode}
\Phi''=\left(k_x^{2}-\mu-\omega+C\right)\Phi-\frac{J}{2\pi}\,\Phi^{3}+\frac{J^{2}}{4\pi^{2}}\,\Phi^{5},
\end{equation}
which governs the travelling-wave profiles in $y$.

Eq.~\eqref{eq:reduced-ode} carries no explicit $y$-dependence, so it integrates once. With $u=(\Phi')^{2}$ one has $\Phi''=\tfrac{1}{2}\,\dd u/\dd\Phi$, and integration with respect to $\Phi$ gives the polynomial first integral
\begin{equation}\label{eq:first-integral}
(\Phi')^{2}=a+\left(k_x^{2}-\mu-\omega+C\right)\Phi^{2}-\frac{J}{4\pi}\,\Phi^{4}+\frac{J^{2}}{12\pi^{2}}\,\Phi^{6},
\end{equation}
in which $a$ is the single integration constant; the derivation is given in Appendix~\ref{sec:Solution-CHJ}, Step~1.

\subsection{Travelling-Wave Family}
\label{subsec:1d_solution}

For the family below, the normalisation $\int_{-\infty}^{\infty}\Phi^{2}(y)\,\dd y=1$ is understood per period on the periodic branch. The coefficients of Eq.~\eqref{eq:first-integral}, matched in Appendix~\ref{sec:Solution-CHJ}, are
\begin{equation}
A\equiv k_{x}^{2}-\mu-\omega+C,
\qquad
B\equiv\frac{J}{2\pi}>0.
\label{coeff-AB}
\end{equation}
The periodic and localised family considered below is restricted to the window $0<A<3/16$ (Appendix~\ref{sec:Solution-CHJ}, Step~6). In the squared profile $\Psi=\Phi^{2}$, the first integral takes the standard elliptic form
\begin{equation}
(\Psi_{y})^{2}=P_{4}(\Psi)=4\,\Psi\,g(\Psi),
\qquad
g(\Psi)=a+A\,\Psi-\frac{B}{2}\,\Psi^{2}+\frac{B^{2}}{3}\,\Psi^{3}.
\label{Psi-elliptic}
\end{equation}
The quartic has positive leading coefficient $4B^{2}/3$. On the open periodic branch, its four roots are real and ordered as
\begin{equation}
0=\Psi_{1}<\Psi_{2}<\Psi_{3}<\Psi_{4},
\qquad
P_{4}(\Psi)=\frac{4B^{2}}{3}\,(\Psi-\Psi_{1})(\Psi-\Psi_{2})(\Psi-\Psi_{3})(\Psi-\Psi_{4}).
\label{eq:parameter-roots}
\end{equation}
Hence the orbit is of genus one and occupies $\Psi\in[\Psi_{2},\Psi_{3}]$: $\Psi_{2}$ is the squared background, $\Psi_{3}$ the squared peak at which $\Phi_{y}=0$, while $\Psi_{4}$ never appears in the profile but fixes both the modulus and the frequency. Matching the coefficients of $\Psi^{3}$, $\Psi^{2}$ and $\Psi$ in Eq.~\eqref{eq:parameter-roots} against $4\Psi\,g(\Psi)$ gives
\begin{equation}
\Psi_{2}+\Psi_{3}+\Psi_{4}=\frac{3}{2B},
\qquad
\Psi_{2}\Psi_{3}+\Psi_{2}\Psi_{4}+\Psi_{3}\Psi_{4}=\frac{3A}{B^{2}},
\qquad
a=-\frac{B^{2}}{3}\,\Psi_{2}\Psi_{3}\Psi_{4}.
\label{eq:vieta-main}
\end{equation}
For $a>0$ the orbit of Eq.~\eqref{eq:first-integral} crosses $\Phi=0$, so $|W_{e}|$ has no positive lower bound; the integration constant therefore ranges over $a\in[a_{\mathrm{crit}},0]$, with $a_{\mathrm{crit}}$ given by Eq.~\eqref{eq:a-critical} of Appendix \ref{sec:Solution-CHJ}.

Two dimensionless combinations of the roots label the orbit. The elliptic modulus $m$ and the characteristic $\nu$ of the third kind are
\begin{equation}
m=\frac{\Psi_{4}(\Psi_{3}-\Psi_{2})}{\Psi_{3}(\Psi_{4}-\Psi_{2})},
\qquad
\nu=1-\frac{\Psi_{2}}{\Psi_{3}},
\label{eq:modulus-nu}
\end{equation}
and the frequency and period are
\begin{equation}
\kappa=\sqrt{-a\,\frac{\Psi_{4}-\Psi_{2}}{\Psi_{2}\Psi_{4}}}
=B\sqrt{\frac{\Psi_{3}\,(\Psi_{4}-\Psi_{2})}{3}},
\qquad
T=\frac{2\,\K(m)}{\kappa}.
\label{eq:kappa-T}
\end{equation}
The two parameters satisfy
\begin{equation}
m-\nu=\frac{\Psi_{2}\,(\Psi_{3}-\Psi_{2})}{\Psi_{3}\,(\Psi_{4}-\Psi_{2})}>0
\label{eq:m-minus-nu}
\end{equation}
throughout the open family. Equality $m=\nu$ holds only at the degenerate endpoints. The limits $m\to0$ and $m\to1$ give extended periodic waves and localised profiles, respectively.

Integration of Eq.~\eqref{Psi-elliptic} then yields the closed profile
\begin{equation}
\Phi(y)=\left[\frac{\Psi_{2}}{1-\nu\,\sn^{2}\bigl(\kappa\,(y-y_{0}),m\bigr)}\right]^{1/2}.
\label{eq:Phi-profile}
\end{equation}
The trough $\Phi=\sqrt{\Psi_{2}}$ falls at $y=y_{0}$ and the crest $\Phi=\sqrt{\Psi_{3}}$ half a period above it, at $\kappa\,(y-y_{0})=\K(m)$. The complex field is
\begin{equation}
W_{e}(x,y,t)=\sqrt{\frac{J}{2\pi}}\,
\mathrm{e}^{\mathrm{i}(k_{x}x-\omega t)}\,\Phi(y),
\label{W-full}
\end{equation}
with $\Phi$ given by Eq.~\eqref{eq:Phi-profile}.

The endpoint $a=0$ of the family is reached when the background vanishes, $\Psi_{2}=0$, at which $g$ is quadratic. Writing $s=\sqrt{1-16A/3}$, the two remaining roots are $\Psi_{3}=\tfrac{3(1-s)}{4B}$ and $\Psi_{4}=\tfrac{3(1+s)}{4B}$, the second expression for the frequency in Eq.~\eqref{eq:kappa-T} reduces to $\kappa=\sqrt{A}$, and Eq.~\eqref{eq:Phi-profile} becomes the localised profile
\begin{equation}
\Phi_{\mathrm{sol}}(y)=\left[\frac{4A}{B\left(1+s\cosh\bigl(2\sqrt{A}\,(y-y_{s})\bigr)\right)}\right]^{1/2},
\qquad
\Phi_{\mathrm{sol}}(\pm\infty)=0,
\label{Phi-soliton}
\end{equation}
peaked at $\Phi_{\mathrm{sol}}^{2}(y_{s})=\Psi_{3}=\tfrac{3(1-s)}{4B}$. In the joint small-amplitude limit $B\to0$ and $A\sim B^{2}/16$, one has $s\to1$ and $\Phi_{\mathrm{sol}}^{2}(y)\sim(2A/B)\operatorname{sech}^{2}\bigl(\sqrt{A}(y-y_{s})\bigr)$. The endpoint is approached continuously from the periodic branch: as $a\to0^{-}$ one has $\Psi_{2}=-a/A+O(a^{2})\to0$, $\nu\to1$, $m\to1$ and $T\to\infty$, and Eq.~\eqref{eq:Phi-profile} converges pointwise to Eq.~\eqref{Phi-soliton}. The minimum amplitude scales as $\sqrt{\Psi_{2}}\sim\sqrt{-a/A}$.

For the periodic branch the normalisation is imposed per period, and Eq.~\eqref{eq:period-integral} of Appendix \ref{sec:Solution-CHJ} reduces it to the single relation
\begin{equation}
\frac{2\,\Psi_{2}}{\kappa}\int_{0}^{\K(m)}\frac{\dd u}{1-\nu\,\sn^{2}(u,m)}=1
\label{eq:norm-eq}
\end{equation}
between $A$, $a$ and $B$. For the localised endpoint the same condition applied to Eq.~\eqref{Phi-soliton} gives
\begin{equation}
\int_{-\infty}^{\infty}\Phi_{\mathrm{sol}}^{2}(y)\,\dd y
=\frac{\sqrt{3}}{2B}\,\ln\frac{3-16A}{\bigl(\sqrt{3}-4\sqrt{A}\bigr)^{2}}=1.
\label{eq:parameter-2}
\end{equation}
Appendix~\ref{sec:Solution-CHJ}, Step~6 solves Eq.~\eqref{eq:parameter-2} in closed form through Eq.~\eqref{eq:soliton-closed-form},
\begin{equation}
A=\frac{3}{16}\,\tanh^{2}\frac{B}{\sqrt{3}}
=\frac{3}{16}\,\tanh^{2}\frac{J}{2\pi\sqrt{3}},
\qquad
B=\sqrt{3}\,\mathrm{artanh}\sqrt{\frac{16A}{3}},
\label{eq:parameter-3}
\end{equation}
the two branches being inverse: the normalisation is a bijection between $A\in\bigl(0,\tfrac{3}{16}\bigr)$ and $B\in(0,\infty)$, so each $J>0$ fixes exactly one normalised localised profile.

Table~\ref{tab:params} summarises the parameters.
\begin{table}[hbtp]
\small
\centering
\caption{Travelling-wave parameters. The last column gives the equation in the text, or the step of Appendix~\ref{sec:Solution-CHJ}, at which each entry is derived.}
\label{tab:params}
\begin{tabular}{@{}>{\raggedright\arraybackslash}p{0.235\textwidth}>{\raggedright\arraybackslash}p{0.615\textwidth}>{\raggedright\arraybackslash}p{0.075\textwidth}@{}}
\toprule
Quantity & Expression & Where\\
\midrule
Linear coefficient $A$ & $A=k_x^2 - \mu - \omega + C$ & Eq.~\eqref{coeff-AB} \\
Nonlinear coefficient $B$ & $B=\dfrac{J}{2\pi}$ & Eq.~\eqref{coeff-AB} \\
Roots of $g$ & $0=\Psi_1<\Psi_2<\Psi_3<\Psi_4$; bounded orbit $\Psi\in[\Psi_2,\Psi_3]$ & Step 2 \\
Integration constant $a$ & $a=-\dfrac{B^{2}}{3}\,\Psi_2\Psi_3\Psi_4\in[a_{\mathrm{crit}},\,0]$ & Steps 3--4 \\
Elliptic modulus $m$ & $m=\dfrac{\Psi_4(\Psi_3-\Psi_2)}{\Psi_3(\Psi_4-\Psi_2)}$ & Step 5 \\
Third-kind characteristic $\nu$ & $\nu=1-\dfrac{\Psi_2}{\Psi_3}$, \quad $m-\nu=\dfrac{\Psi_2(\Psi_3-\Psi_2)}{\Psi_3(\Psi_4-\Psi_2)}>0$ & Step 5 \\
Frequency $\kappa$ & $\kappa=\sqrt{-a\,\dfrac{\Psi_4-\Psi_2}{\Psi_2\Psi_4}}=B\sqrt{\dfrac{\Psi_3(\Psi_4-\Psi_2)}{3}}$ & Step 5 \\
Period $T$ & $T=\dfrac{2\,\K(m)}{\kappa}$ & Step 5 \\
Real profile $\Phi(y)$ & $\Phi(y)=\left[\dfrac{\Psi_2}{1-\nu\,\sn^{2}\bigl(\kappa\,(y-y_0),m\bigr)}\right]^{1/2}$ & Step 5 \\
Normalisation, periodic & $\dfrac{2\,\Psi_2}{\kappa}\displaystyle\int_{0}^{\K(m)}\dfrac{\dd u}{1-\nu\,\sn^{2}(u,m)}=1$ & Step 6 \\
Localised endpoint & $a=0$, \ $s=\sqrt{1-\dfrac{16A}{3}}$, \ $\kappa=\sqrt{A}$, \ $\Phi_{\mathrm{sol}}(y)=\left[\dfrac{4A}{B\left(1+s\cosh\bigl(2\sqrt{A}\,(y-y_s)\bigr)\right)}\right]^{1/2}$ & Step 6 \\
Normalisation, localised & $A=\dfrac{3}{16}\,\tanh^{2}\dfrac{J}{2\pi\sqrt{3}}$, bijective in $J\in(0,\infty)$ & Step 6 \\
Complex field $W_e$ & $W_e(x,y,t)=\sqrt{\dfrac{J}{2\pi}}\,e^{i(k_x x - \omega t)}\,\Phi(y)$ & Eq.~\eqref{eq:tw-ansatz} \\
Window on $A$ & $0<A<\dfrac{3}{16}$, required by the localised endpoint & Eq.~\eqref{eq:parameter-3} \\
\bottomrule
\end{tabular}
\end{table}

Fig.~\ref{fig:fig-1} illustrates four members of the family at $A=0.04$, labelled by the characteristic $\nu=0.2,\,0.5,\,0.8$ and by the localised endpoint $\nu=1$, in the four views and the shape-function panel (q) itemised in the caption. The window is $x\in[-12.5,12.5]$ and $y\in[-27,27]$. The medium and carrier are $\mu=1.0$, $k_x=0.8$, $C=1.0$, $\omega=0.6$. Each row is normalised per period, so Eq.~\eqref{eq:norm-eq} fixes its own total intensity $J=2\pi B$; the four rows then have elliptic modulus $m=0.2054$, $0.5100$, $0.8080$, $1$, frequency $\kappa=0.1460$, $0.1606$, $0.1809$, $0.2000$, period $T=22.765$, $23.192$, $25.157$, $\infty$ and total intensity $J=5.9538$, $5.9231$, $5.8029$, $5.4381$.
\begin{figure}[!htbp]
    \centering
    \includegraphics[width=0.82\textwidth, height=0.71\textheight, keepaspectratio]{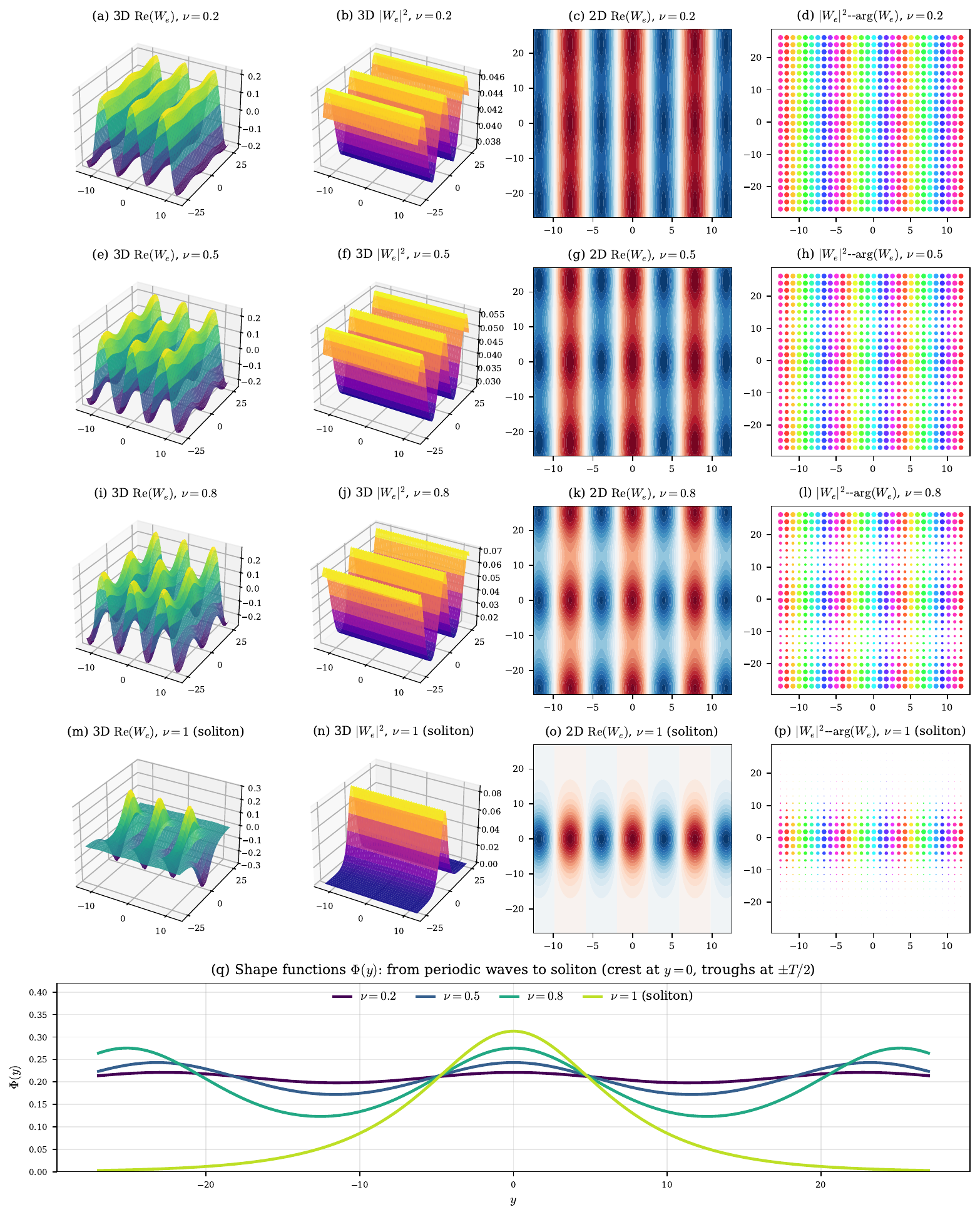}
\caption{
Spatial visualisation of the travelling-wave solutions Eq.~\eqref{eq:Phi-profile}, at $t=0$, over the window quoted in the text above. Panels (a)--(d): $\nu=0.2$; (e)--(h): $\nu=0.5$; (i)--(l): $\nu=0.8$; (m)--(p): the localised endpoint $\nu=1$ of Eq.~\eqref{Phi-soliton}; within each row the 3D surface of $\mathrm{Re}(W_e)$, the 3D surface of $|W_e|^{2}$, the filled contour of $\mathrm{Re}(W_e)$, and the intensity as marker area with $\arg W_e$ as colour. Panel (q): the four shape functions, crest at the origin ($y_0=-T/2$ in Eq.~\eqref{eq:Phi-profile}), troughs at $\pm T/2$.
}
    \label{fig:fig-1}
\end{figure}

\section{Probability Measure and Its Statistics}
\label{sec:Contact-Probability-Functional}

\subsection{Probability Measure}
\label{sec:Probability_Measure}
\begin{theorem}[On-Shell Probability Density]
\label{thm:reciprocal-density}
On the travelling-wave submanifold $\mathcal{T}$, with $\inf\Phi>0$, the probability density determined by the on-shell action $\mathcal{S}$ satisfies
\begin{equation}
\mathcal{P}[W]\;\propto\;\frac{1}{|W_e|^2}.
\label{eq:reciprocal-density}
\end{equation}
\end{theorem}

\begin{proof}
The proof is given in Appendix~\ref{app:proof-reciprocal}. 
\end{proof}

Here $W_e$ is the on-shell representative used to evaluate the weight: a configuration is labelled by a field of the CGLE Eq.~\eqref{eq:qcgle}, whereas the canonical transformation Eq.~\eqref{eq:direct_transformation} selects this representative on $\mathcal{T}$. On $\mathcal{T}$, Theorem~\ref{thm:reciprocal-density} follows from the branch-summed action of $W_e$ and $W_e^{*}$: the radial chain rule Eq.~\eqref{eq:chain-rule-radial} yields Eq.~\eqref{eq:S-two-branch}. For the one-dimensional complex fibre, the radial coordinate element is $\dd\mu_{\mathrm{red}}=\dd|W_e|$, and substitution of the on-shell action gives
\begin{equation}
P[|W|; \dd |W|] = \mathcal{P}[|W|]\,\dd \mu_{\mathrm{red}}[|W|]
\propto \exp\left(-2\ln |W_e|  \right)\,\dd |W_e|
= |W_e|^{-2}\,\dd |W_e| .
\label{eq:radial-measure-two-branch}
\end{equation}

\subsection{Intensity Constraint and Free Energy}
\label{sec:cont_norm}

Since the CHJ reduction treats $J$ as a fixed external parameter on the statistical steady-state time window, the term $\lambda J$ is introduced to control the weighting between intensity sectors while remaining constant within each fixed-$J$ sector, so the fixed-intensity condition $\delta\mathcal{S}=0$ is unaffected.
The statistical behaviour of the system is described by the measure conditioned on the intensity,
\begin{equation}
\mathcal{P}[W|\lambda] = \frac{1}{\mathcal{Z}(\lambda)} \exp\!\Bigl( \mathcal{S}[W] + \lambda J[W] \Bigr),
\label{eq:cond_prob}
\end{equation}
where
\begin{equation}\label{eq:J-normalization}
J[W] = \int_{0}^{2\pi}\!\dd x\int\dd y\; |W(\vb{r})|^{2}
\end{equation}
is the total intensity. There the $x$-range of length $2\pi$ and the $y$-profile normalisation $\int\Phi^{2}\dd y=1$ give $|W|^{2}=\tfrac{J}{2\pi}\Phi^{2}$, with $J=2\pi B$ by Eq.~\eqref{coeff-AB}. The multiplier $\lambda$ is conjugate to $J$.

Because the exponential in Theorem~\ref{thm:measure} is strictly positive, stationarity of $\mathcal{P}$ can be examined through $\mathcal{S}$ and the intensity constraint. Variations in the reduced variables of $\mathcal{T}$, with the intensity controlled by $\lambda$ or held fixed, give
\begin{equation}
\delta\mathcal{S} = -\lambda\,\delta J
\qquad\text{and, at fixed }J,\qquad
\delta\mathcal{S} = 0 ,
\label{eq:two-stage-variational}
\end{equation}
the second being the fixed-intensity condition used in the CHJ reduction.

On the travelling-wave submanifold, the pointwise effective action is minus the on-shell action of Eq.~\eqref{eq:S-two-branch}:
\[
S_{\text{eff}}[W] = 2\ln\!\left(\frac{|W_e(\mathbf{r})|}{|W_e|_s}\right),
\]
and the comparative function (or ``free energy'') becomes
\begin{equation}
\mathcal{F} = S_{\text{eff}} - \lambda J,
\label{eq:comparative_func}
\end{equation}
which on the travelling-wave submanifold depends on a single shape parameter, the third-kind characteristic $\nu$ (Eq.~\eqref{eq:modulus-nu}).

On $\mathcal{T}$, the relation $|W_e|^{2}=B\Phi^{2}$ with $B=J/2\pi$ separates the effective action at fixed shape into $\ln J$ and amplitude-independent terms. Varying the amplitude gives
\begin{equation}
\frac{\partial \mathcal{F}}{\partial J} = 0 \quad \Longrightarrow \quad \lambda = \frac{\partial S_{\text{eff}}}{\partial J} = \frac{1}{J}.
\label{eq:lambda_relation}
\end{equation}
Since $\partial^{2}\mathcal{F}/\partial J^{2}=-1/J^{2}<0$, the stationary value maximises $\mathcal{F}$ and hence minimises $\mathcal{P}\propto\exp(-\mathcal{F})$. The Vieta relations and normalisation make $J$ a function of $\nu$ once $A$ is fixed by the medium and the carrier.

The partition function $\mathcal{Z}(\lambda)$ is a radial integral. On the spatially uniform amplitude direction $W=R\,e^{i\theta}$, the intensity is $I=R^{2}$, and the cell of Eq.~\eqref{eq:J-normalization}, of area $2\pi T$, gives $J=2\pi T I$; hence
\[
\mathcal{Z}(\lambda)=\int\exp\!\bigl[-2\ln(R/|W_e|_s)+2\pi T\lambda R^{2}\bigr]\dd R ,
\]
The integral runs between a finite-volume cut-off $R_{\min}$ and a saturation cut-off $R_{\max}$, and converges for every $\lambda>0$. The amplitude condition Eq.~\eqref{eq:lambda_relation} then fixes the stationary amplitude,
\begin{equation}
R_{\lambda} = \sqrt{\frac{J}{2\pi T}} = \frac{1}{\sqrt{2\pi T\,\lambda}} \qquad (\lambda > 0).\label{eq:R_lambda}
\end{equation}

\subsection{Statistics}
\label{sec:Heavy_Tail_Scaling}

\begin{figure}[!b]
\centering
\includegraphics[width=0.98\textwidth]{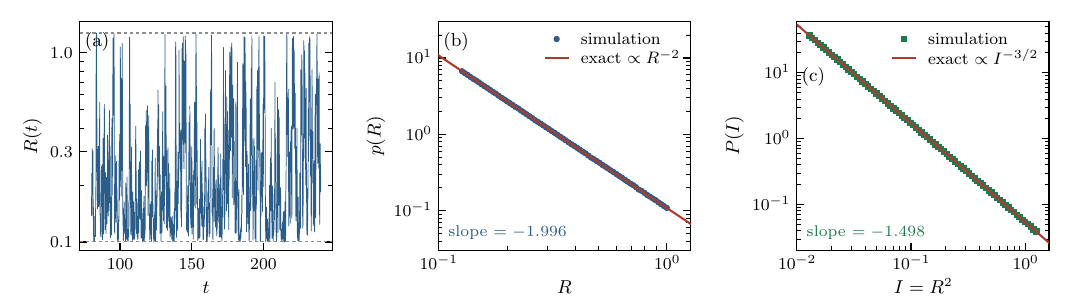}
\caption{Stochastic radial realisation of the measure. (a) One trajectory of Eq.~\eqref{eq:G-sde} (Euler--Maruyama, mirror reflection at the dashed cut-offs, $\Delta t=2\times10^{-4}$); the excursions are intermittent. (b) Stationary line density of the same run against the exact $\mathcal{N}R^{-2}$ with $\mathcal{N}=[1/R_{\min}-1/R_{\max}]^{-1}$; the least-squares slope over the interior of the domain is $-1.996$. (c) Stationary intensity density against the exact $\mathcal{N}I^{-3/2}/2$, the exponent of Eq.~\eqref{eq:IntensityDist}; the least-squares slope over the interior of the domain is $-1.498$. Parameters $\mu=1$, $V=100$, giving $R_{\min}=0.1$ and $R_{\max}=1.272$ through the cut-offs of Section~\ref{sec:cont_norm}; sample size and further diagnostics in Appendix~\ref{app:stoch-radial}.}
\label{fig:fig-2}
\end{figure}

For the on-shell intensity $I=R^{2}=|W_e|^{2}$, changing variables in the reduced radial measure $\dd\mu_{\mathrm{red}}=\dd R$ transforms the radial density $\mathcal{P}\propto R^{-2}$ into
\begin{equation}
\mathcal{P}[I] = \mathcal{P}[|W|]\,\frac{\dd|W_e|}{\dd I} \propto |W_e|^{-2}\cdot I^{-1/2} = I^{-3/2},
\label{eq:IntensityDist}
\end{equation}
valid in the asymptotic range $I_{\min}\ll I\ll I_{\max}$, with $I_{\min}\sim V^{-1}$ set by finite size or background noise and $I_{\max}=\mathcal{N}_{+}$ by saturation of the quintic term (Eq.~\eqref{eq:N-uniform}). Appendix~\ref{app:stoch-radial} samples this weight with a radial diffusion driven by the CGLE drift and with a Metropolis chain on the shape parameter; the sampling verifies the reduction to Eq.~\eqref{eq:IntensityDist}, the moment scaling and the absence of estimator bias (Fig.~\ref{fig:fig-2}). The appendix also shows that isotropic constant positive noise with the pure CGLE drift cannot carry this law.

At finite cut-offs the intermediate $I^{-3/2}$ law has finite variance, which diverges as $I_{\max}\to\infty$. This power law differs from the non-Gaussian, enhanced large-amplitude probabilities reported by Kim and Ott \cite{Kim2003RareEvents}: their approximately stretched-exponential distribution $P(|A|)\approx e^{-|A|^{\eta}}$ with $\eta<1$ decays faster than any power law.

The $p$-th order moment of the normalised density $\mathcal{P}[I]=\tfrac{\mathcal{N}}{2}I^{-3/2}$, with $\mathcal{N}=[I_{\min}^{-1/2}-I_{\max}^{-1/2}]^{-1}$, is
\begin{equation}
\langle I^{p}\rangle = \int_{I_{\min}}^{I_{\max}} I^{p}\,\mathcal{P}[I]\,\dd I = \frac{\mathcal{N}}{2}\int_{I_{\min}}^{I_{\max}} I^{p-3/2}\,\dd I.
\label{eq:Moments}
\end{equation}
As the corresponding cut-off is removed, the unnormalised integral diverges at the lower limit for $p<1/2$ (small-amplitude infrared singularity) and at the upper limit for $p>1/2$ (large-amplitude ultraviolet singularity). At the critical order $p=\tfrac12$, the integrand is $I^{-1}$ and $\langle I^{1/2}\rangle=\tfrac{\mathcal{N}}{2}\ln(I_{\max}/I_{\min})$ is logarithmic in the cut-off ratio, remaining finite at finite cut-offs. With these cut-offs, the normalisation $\mathcal{N}\sim I_{\min}^{1/2}$ cancels the infrared divergence and gives $\langle I^{p}\rangle\sim I_{\min}^{p}$ for $p<1/2$, yielding
\begin{equation}
\langle I^{p}\rangle \sim 
\begin{cases}
I_{\max}^{\,p-1/2}, & p>\tfrac{1}{2},\\[6pt]
\tfrac{\mathcal{N}}{2}\,\ln(I_{\max}/I_{\min}), & p=\tfrac{1}{2},\\[6pt]
V^{-p}, & p<\tfrac{1}{2}.
\end{cases}
\label{eq:MomentScaling}
\end{equation}
For $p\geq 1$, the saturation cut-off $I_{\max}$, rather than an intrinsic correlation length, controls the moments.

The structure functions $S_{p}(\ell)=\langle|W(\mathbf{r}+\boldsymbol\ell)-W(\mathbf{r})|^{p}\rangle$ are two-point statistics characterising spatial intermittency, whereas the contact measure is a single-point radial density $P(|W|)\propto|W|^{-2}$. The increment $|W_2-W_1|$ is bounded by the triangle inequality,
\[
\big||W_2|-|W_1|\big|\leq|W_2-W_1|\leq|W_2|+|W_1|.
\]
Since $|W_1|$ and $|W_2|$ share the regularly varying tail $P(|W|>x)\sim x^{-1}$ of index $\alpha=1$, both bounds have the same tail exponent: the sum by subexponentiality, the difference by the dominance principle whenever the field varies in space. Hence $|\delta W|$ is regularly varying with the same index, and for $p>1$ the structure function shares the ultraviolet scaling of the marginal moment,
\begin{equation}
S_{p} \sim \langle |W|^p \rangle = \langle I^{p/2} \rangle \sim I_{\max}^{(p-1)/2},\qquad p>1.
\label{eq:StructureFunction}
\end{equation}
The exponents $\zeta_{p}=(p-1)/2$ are anomalous and arise from the singularity of the reciprocal on-shell density at $|W_e|\to 0$.

The decay $\mathcal{P}\propto I^{-3/2}$ can be compared with the $q$-exponential distribution $e_{q}(-\beta I)$ of non-extensive statistical mechanics: its asymptotic form $e_{q}(x)\sim|x|^{-1/(q-1)}$ for $q>1$ matches the power law at $q=5/3$, and the role of $q$ in low-dimensional dissipative systems is documented in \cite{Lyra1998Nonextensivity}. Within this interpretation, the constraint $\lambda$ plays the role of an inverse temperature. For the radial process of Appendix~\ref{app:stoch-radial}, maximum-likelihood fitting of $e_{q}(-\beta I)$ to the simulated intensity density gives $\hat q=1.667$ (95\% CI $[1.666,1.668]$), within $0.03\%$ of the value $5/3$ fixed by the exponent $-3/2$. Cut-off sweeps of the same process reproduce Eq.~\eqref{eq:MomentScaling} (Fig.~\ref{fig:fig-3}), with the simulated values following the exact radial moment $\mathcal{N}(R_{\max}^{2p-1}-R_{\min}^{2p-1})/(2p-1)$.

\begin{figure}[!htbp]
    \centering
    \includegraphics[width=0.72\textwidth]{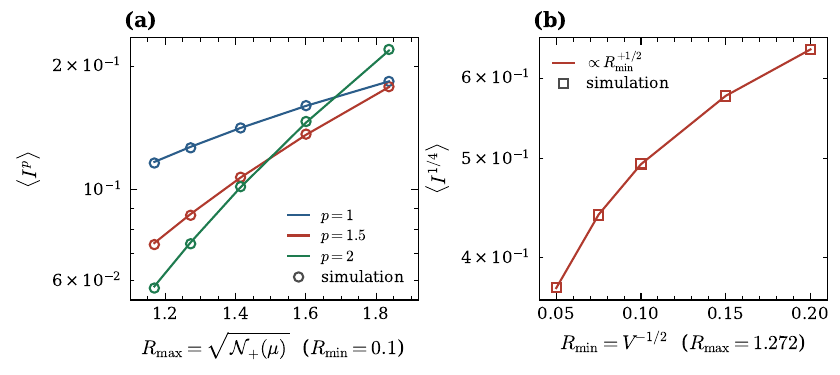}
\caption{Cut-off scaling of the intensity moments on the stochastic realisation of Appendix~\ref{app:stoch-radial} (the process of Fig.~\ref{fig:fig-2}), verifying Eq.~\eqref{eq:MomentScaling}. (a) Ultraviolet: $\langle I^{p}\rangle$ for $p=1,1.5,2$ against the saturation cut-off $R_{\max}=\sqrt{\mathcal{N}_{+}(\mu)}$ at $R_{\min}=0.1$, swept through $\mu\in[0.5,8]$; open symbols are simulated, lines the exact radial moment $\propto R_{\max}^{2p-1}$. (b) Infrared: $\langle I^{1/4}\rangle$ against the finite-volume cut-off $R_{\min}=V^{-1/2}$ at $R_{\max}=1.272$, the exact moment $\propto R_{\min}^{1/2}$.}
    \label{fig:fig-3}
\end{figure}

\subsection{Periodic and Localised States}
\label{sec:Phase_Diagram_m}

\subsubsection{Landscape of the Measure}

On the family, the landscape $\mathcal{F}$ has two competing terms. The field average $S_{\text{eff}}=\langle\ln|W_e|^{2}\rangle_{y}$ is, up to an additive constant, minus the period-averaged logarithm of the reciprocal on-shell density; its logarithmic divergence at a vanishing profile trough favours the localised state. The intensity constraint $\lambda J$ favours the extended end, where $J=2\pi B$ is largest.

At fixed $A$ the family is parametrised by the third-kind characteristic $\nu=1-\Psi_{2}/\Psi_{3}$ of Eq.~\eqref{eq:modulus-nu}, the modulation depth, which runs from the constant-amplitude state at $\nu=0$ to the localised limit $\nu=1$, where the trough vanishes (Eq.~\eqref{Phi-soliton}). Introduce the scaled roots
\begin{equation}
x\equiv B\Psi_{3},
\qquad
y\equiv B\Psi_{4},
\qquad
z\equiv B\Psi_{2}=x(1-\nu),
\label{eq:xyz-def}
\end{equation}
in terms of which the first two Vieta relations Eq.~\eqref{eq:vieta-main} read $x+y+z=\tfrac32$ and $xy+yz+zx=3A$, neither involving $B$. Eliminating $y$ and $z$ leaves a quadratic for $x$, whose branch is fixed by $x\to\tfrac12\bigl(1-\sqrt{1-4A}\bigr)$ as $\nu\to0$:
\begin{equation}
\begin{split}
x(A,\nu)&=\frac{\tfrac32(2-\nu)-\sqrt{\Delta_{\nu}}}{2(\nu^{2}-3\nu+3)},
\\
\Delta_{\nu}&\equiv\frac94(2-\nu)^{2}-12A(\nu^{2}-3\nu+3),\\
y&=\frac32-x(2-\nu).
\end{split}
\label{eq:xnu}
\end{equation}
The modulus, the frequency, the amplitude scale and the total intensity then follow from Eq.~\eqref{eq:xyz-def} in closed form,
\begin{equation}
\begin{gathered}
m=\frac{y\,\nu}{y-x(1-\nu)},
\qquad
\kappa=\sqrt{\frac{x\,\bigl[y-x(1-\nu)\bigr]}{3}},
\\[5pt]
B=\frac{2x(1-\nu)}{\kappa}\int_{0}^{\K(m)}\frac{\dd u}{1-\nu\,\sn^{2}(u,m)},
\qquad
J(\nu)=2\pi B .
\end{gathered}
\label{eq:J_of_m_exact}
\end{equation}
The Vieta relations Eq.~\eqref{eq:vieta-main} and the normalisation Eq.~\eqref{eq:norm-eq} thus determine $B$, and with it the amplitude scale, from the shape: at fixed $A$ the family is one-dimensional, and $\mathcal{F}=\mathcal{F}(\nu,\lambda)$.

For the algebraic domain, $\Delta_{\nu}\ge0$ in Eq.~\eqref{eq:xnu} requires $A\le3(2-\nu)^{2}/\bigl[16(\nu^{2}-3\nu+3)\bigr]$, whose right-hand side decreases from $\tfrac14$ at $\nu=0$ to $\tfrac3{16}$ at $\nu=1$. Throughout the window $0<A<3/16$ the bound is automatic, and the localised endpoint $\nu=1$ is reached as a limit for every admissible $A$. The infrared weight then sets the domain cap: $|W_{e}|^{2}=B\Psi$ attains its minimum $B\Psi_{2}=x(1-\nu)$ at the trough of Eq.~\eqref{eq:Phi-profile}, and bounding the amplitude below by $R_{\min}$ with $R_{\min}^{2}=I_{\min}\sim V^{-1}$ gives
\begin{equation}
\begin{gathered}
0\le\nu\le\nu_{\max},
\qquad
x(A,\nu_{\max})\,\bigl(1-\nu_{\max}\bigr)=R_{\min}^{2}=\frac{1}{V},
\\[4pt]
\nu_{\max}=1-\frac{4}{3(1-s)}\,\frac1V+\mathcal{O}\!\left(V^{-2}\right),
\qquad
s\equiv\sqrt{1-\frac{16A}{3}} ,
\end{gathered}
\label{eq:IR-domain}
\end{equation}
the second line following from $x\to3(1-s)/4$ as $\nu\to1$. At the upper end the trough intensity $x(1-\nu)$ tends to zero and the weight diverges as $(1-\nu)^{-1/2}$ (Appendix~\ref{app:E6}). The infrared cutoff keeps the trough at $|W_{e}|_{\min}=R_{\min}>0$, so $\inf\Phi>0$ holds on the truncated family $0\le\nu\le\nu_{\max}<1$ and Theorem~\ref{thm:reciprocal-density} applies on it without exception. The uncut endpoint $\nu\to1$, where $\inf\Phi\to0$, lies outside the hypothesis of the theorem but is approached only through the integrable divergence above, so the normalisation remains finite as $\nu_{\max}\to1$.

Period-averaging $\ln|W_e|^{2}$ (Appendix~\ref{app:transition}) makes the pointwise $S_{\text{eff}}$ a function on the family; the reference amplitude $|W_e|_s$ contributes only an additive constant that cancels in $\mathcal{F}'$,
\begin{equation}
\begin{gathered}
S_{\text{eff}}(\nu)=\bigl\langle\ln|W_e|^{2}\bigr\rangle_{y}
=\ln x+\ln(1-\nu)-G(\nu,m),
\\[5pt]
G(\nu,m)\equiv\frac{1}{2\K(m)}\int_{0}^{2\K(m)}\ln\!\bigl(1-\nu\,\sn^{2}(u,m)\bigr)\,\dd u ,
\end{gathered}
\label{eq:Seff-nu}
\end{equation}
Appendix~\ref{app:E2} separates $G$ into $\tfrac12\ln(1-m)$ and a nonsingular average on the family ($m>\nu$), leaving the infrared localisation term $\tfrac12\ln(1-\nu)$.

\subsubsection{Bifurcation and Crossover}

The extended end $\nu=0$ is a stationary point of $\mathcal{F}$ for every $\lambda$, because both derivatives vanish identically:
\begin{equation}
S_{\text{eff}}'(0)=\frac12-1+\frac12=0,
\qquad
J'(0)=0,
\qquad\Longrightarrow\qquad
\mathcal{F}'(0)=0\quad\text{for all }\lambda>0 .
\label{eq:Fprime0}
\end{equation}
Appendices~\ref{app:E1}--\ref{app:E3} verify both for every admissible $A$. Since $\mathcal{F}'(0)$ vanishes, the threshold is determined by the curvature,
\begin{equation}
\mathcal{F}''(0)=S_{\text{eff}}''(0)-\lambda J''(0)=-J''(0)\,(\lambda-\lambda_c),
\qquad
\lambda_c=\frac{S_{\text{eff}}''(0)}{J''(0)} .
\label{eq:lambda_c}
\end{equation}
Writing $r=\sqrt{1-4A}$, both curvatures are closed forms in $A$,
\begin{equation}
S_{\text{eff}}''(0)=-\frac12\,\frac{r-1+3A}{r-1+4A},
\qquad
J''(0)=\frac{\pi^{2}\,\bigl(5-5r-16rA\bigr)}{48\,(1-4A)\sqrt{r-1+4A}},
\label{eq:curvatures}
\end{equation}
so that $\lambda_c$ is a function of $A$ alone; at $A=0.04$, $\lambda_c=1.594$.

Since $J'(0)=0$ and $J''(0)<0$ (Appendix~\ref{app:E1}), $J'(\nu)<0$ for $0<\nu\ll1$; write
\begin{equation}
\mathcal{F}'(\nu)=J'(\nu)\,\bigl[h(\nu)-\lambda\bigr],
\qquad
h(\nu)\equiv\frac{S_{\text{eff}}'(\nu)}{J'(\nu)}.
\label{eq:h-def}
\end{equation}
Since $\mathcal{F}''(0)=-J''(0)(\lambda-\lambda_c)$ and $J''(0)<0$,
$\nu=0$ is a local minimum of the weight for $\lambda<\lambda_c$ and a local
maximum for $\lambda>\lambda_c$. At $A=0.04$, Appendix~\ref{app:E4} gives
$h(0)=\lambda_c$, $h'(0)=0$ and $h''(0)>0$. Hence $h(\nu)>\lambda_c$ for
$0<\nu\ll1$. For $\lambda=\lambda_c(1+\varepsilon)$ with
$0<\varepsilon\ll1$, the equation $h(\nu)=\lambda$ has the solution
$\nu^*(\lambda)$ of Eq.~\eqref{eq:onset_linear} near $\nu=0$. At this
solution, $\mathcal{F}''(\nu^*)=J'(\nu^*)h'(\nu^*)<0$: $\mathcal{F}$ has an
interior maximum and the weight an interior minimum. Thus an interior
stationary point is born at $\lambda_c$ near $\nu=0$. Since
$\nu=1-\Psi_2/\Psi_3\ge0$, only one side of the bifurcation is available: a
one-sided, or half-branch, bifurcation.

The threshold and the onset law hold only on part of the medium range. Because Eq.~\eqref{eq:third-order} holds identically in $A$, the critical point is degenerate to third order and the expansion of $\mathcal{F}$ begins at fourth order; with $\lambda=\lambda_c(1+\varepsilon)$ the separatrix then obeys the square-root law $\nu^{*}=c\,\varepsilon^{1/2}$ of Eq.~\eqref{eq:onset_linear}, with $c$ real exactly where the fourth-order coefficient $Q$ is negative. At the medium $A=0.04$ used throughout, $Q<0$; Appendix~\ref{app:E4} evaluates $Q$ and $c$ and gives the range of $A$ on which $Q<0$.

The equilibrium distribution is the weight normalised over the capped domain,
\begin{equation}
\mathcal{Z}(\lambda) = \int_{0}^{\nu_{\max}} \exp\!\bigl(-\mathcal{F}(\nu,\lambda)\bigr)\,\dd\nu ,
\label{eq:Z_m_decomp}
\end{equation}
and the equilibrium order parameter and its response are
\begin{equation}
\langle\nu\rangle(\lambda) = \frac{1}{\mathcal{Z}}\int_{0}^{\nu_{\max}} \nu\,\exp\!\bigl(-\mathcal{F}(\nu,\lambda)\bigr)\,\dd\nu ,
\qquad
\chi(\lambda) = -\frac{\dd\langle\nu\rangle}{\dd\lambda}.
\label{eq:order_param}
\end{equation}
Since $\partial_{\lambda}e^{-\mathcal{F}}=J\,e^{-\mathcal{F}}$, differentiating under the integral sign expresses the response as an exact covariance,
\begin{equation}
\begin{gathered}
\frac{\dd\langle O\rangle}{\dd\lambda}=\mathrm{Cov}_{\lambda}\bigl(O,J\bigr)
\quad\text{for any observable }O,
\qquad\text{and hence, with }O=\nu,
\\[5pt]
\chi(\lambda)=-\mathrm{Cov}_{\lambda}\bigl(\nu,J(\nu)\bigr)=\langle\nu\rangle\langle J\rangle-\langle\nu J\rangle .
\end{gathered}
\label{eq:chi-cov}
\end{equation}
Eq.~\eqref{eq:Z_m_decomp} has an analytic integrand on a finite domain; the integrable $(1-\nu)^{-1/2}$ singularity is held away from the endpoint by $\nu_{\max}<1$, so both $\langle\nu\rangle$ and $\chi$ are analytic in $\lambda$. The equilibrium conversion is therefore a smooth crossover, while the bifurcation acts on the stationary points of the landscape.

The centre of the crossover is the value $\lambda_{1/2}$ at which the order parameter passes through one half, $\langle\nu\rangle(\lambda_{1/2})=\tfrac12$, and the width is its distance from the bifurcation threshold,
\begin{equation}
\Delta\lambda\equiv\lambda_{1/2}-\lambda_c .
\label{eq:crossover-def}
\end{equation}
For the parameters of Fig.~\ref{fig:fig-4},
\begin{equation}
\lambda_{1/2}=2.591=1.625\,\lambda_c,
\qquad
\Delta\lambda=0.997 .
\label{eq:crossover_nums}
\end{equation}
The volume enters Eq.~\eqref{eq:Z_m_decomp} only through $\nu_{\max}=1-\mathcal{O}(V^{-1})$, so the two scales converge in the thermodynamic limit. The crossover centre exists only while $\sup_{\lambda}\langle\nu\rangle>\tfrac12$.

The most probable configuration changes at a third scale $\lambda_{\mathrm{mp}}$, where the two endpoint weights are equal. With $w(\nu)=\exp\!\bigl(-\mathcal{F}(\nu,\lambda)\bigr)$, their ratio is
\begin{equation}
\frac{w(\nu_{\max})}{w(0)}
=\exp\!\Bigl[\bigl(S_{\text{eff}}(0)-S_{\text{eff}}(\nu_{\max})\bigr)-\lambda\bigl(J(0)-J(\nu_{\max})\bigr)\Bigr]
=e^{\,4.9432-0.5192\,\lambda},
\label{eq:end-ratio}
\end{equation}
which exceeds unity until $\lambda=4.9432/0.5192=9.521$, so the localised end retains the higher density peak well beyond the mean crossover and after the region $\nu>\tfrac12$ has lost the majority of the probability. Hence
\begin{equation}
\underbrace{\lambda_c\ \text{landscape}}_{1.594}
\;<\;
\underbrace{\lambda_{1/2}\ \text{mean}}_{2.591}
\;<\;
\underbrace{\lambda_{\mathrm{mp}}\ \text{mode}}_{9.521}.
\label{eq:three-scales}
\end{equation}
The infrared divergence of the weight places $\lambda_{\mathrm{mp}}$ far above the other two scales. Fig.~\ref{fig:fig-4} verifies the local bifurcation and the onset law, and displays the three scales.

\begin{figure}[!htbp]
    \centering
    \includegraphics[width=0.80\textwidth, height=0.53\textheight, keepaspectratio]{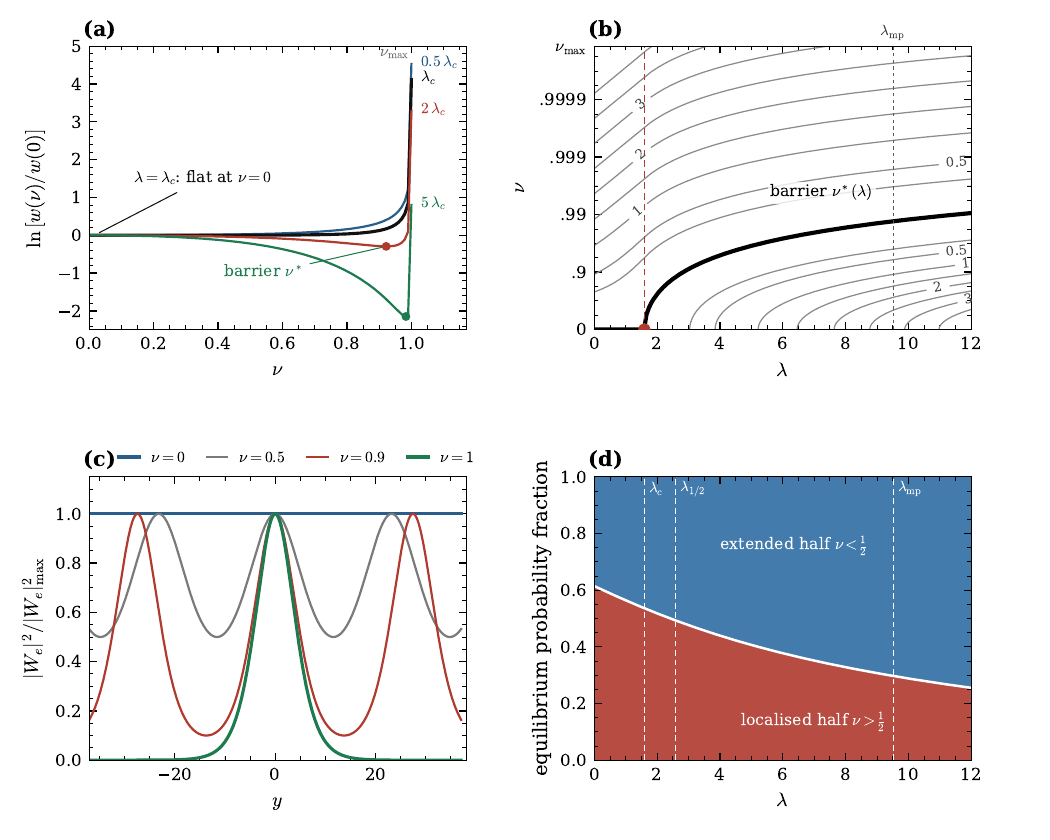}
    \caption{On-shell bifurcation at $A=0.04$ (one choice being $\mu=1.0$, $k_x=0.8$, $C=1.0$, $\omega=0.6$).
    (a) The weight $w(\nu)=\exp(-\mathcal{F})$ relative to its value at the extended end, for $\lambda=0.5\lambda_c,\,\lambda_c,\,2\lambda_c,\,5\lambda_c$.
    (b) Contour lines of the landscape over $(\lambda,\nu)$, the $\nu$ axis spaced by $-\ln(1-\nu)$ so that the infrared end is resolved.
    (c) The states $|W_e|^{2}$ at the crest for $\nu=0$, $0.5$, $0.9$ and the localised limit $\nu=1$ of Eq.~\eqref{Phi-soliton}.
    (d) The equilibrium partition of the probability between the fixed halves $\nu>\tfrac12$ and $\nu<\tfrac12$; $\lambda_c=1.594$, $\lambda_{1/2}=2.591$ and $\lambda_{\rm mp}=9.521$ marked. The partition reaches one half at $\lambda\approx2.42$, the median of $\nu$, below $\lambda_{1/2}$, which marks $\langle\nu\rangle=\tfrac12$.}
    \label{fig:fig-4}
\end{figure}

Independent Metropolis samples agree with quadrature for the density, order parameter, covariance response and volume dependence (Fig.~\ref{fig:fig-5}; method in Appendix~\ref{app:stoch-radial}), the response via Eq.~\eqref{eq:chi-cov}. Mode-locked fibre lasers exhibit a crossover from an extended cnoidal wave to a localised fundamental soliton as tailored intracavity polarisation-dependent loss is swept \cite{Jiang2025FiberLaser}. That experiment provides a qualitative analogue of the smooth equilibrium crossover identified here.

\begin{figure}[!htbp]
    \centering
    \includegraphics[width=0.78\textwidth, height=0.52\textheight, keepaspectratio]{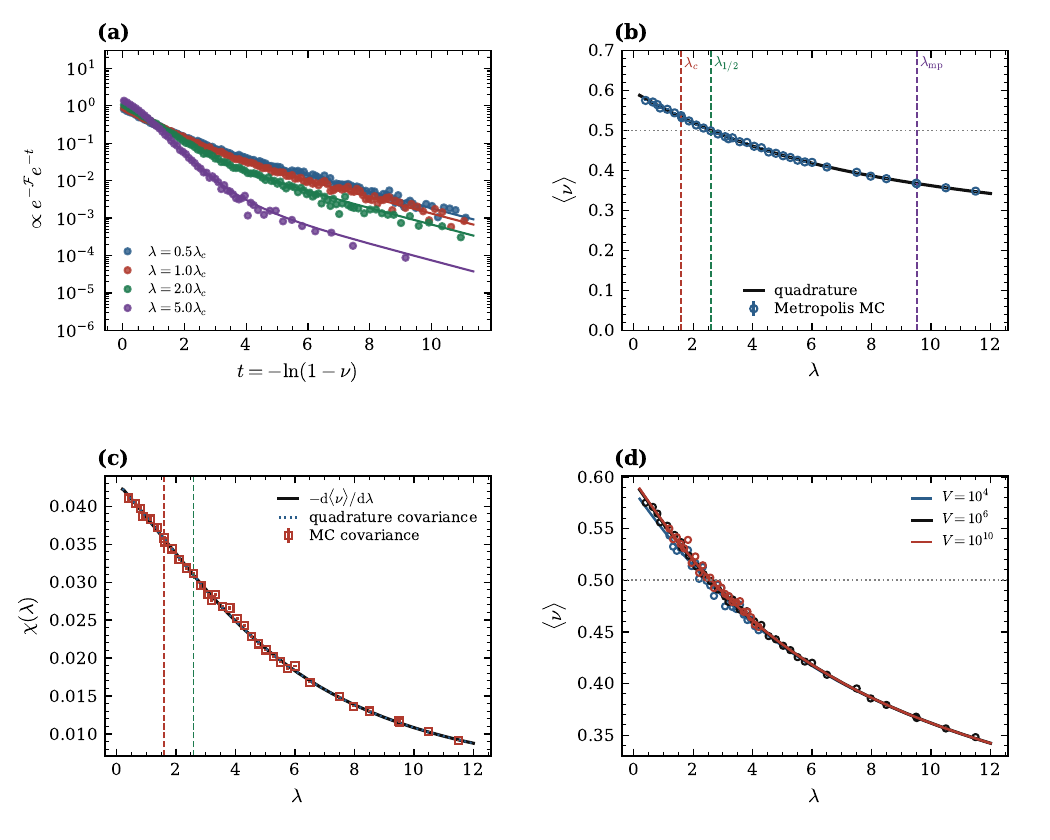}
\caption{Independent Metropolis verification of the on-shell bifurcation at $A=0.04$, sampling the weight $\exp(-\mathcal{F})$ independently of the quadrature used for the order parameter and response. (a) Sampled density of $t=-\ln(1-\nu)$ against the analytic $e^{-\mathcal{F}}e^{-t}$ for $\lambda=0.5\lambda_c,\lambda_c,2\lambda_c,5\lambda_c$. (b) Order parameter $\langle\nu\rangle(\lambda)$ at $V=10^{6}$: Metropolis samples with error bars against the quadrature branch, marking $\lambda_c$, $\lambda_{1/2}$ and $\lambda_{\mathrm{mp}}$. (c) Response $\chi(\lambda)$ from the exact covariance identity Eq.~\eqref{eq:chi-cov}, against $-\dd\langle\nu\rangle/\dd\lambda$ by quadrature. (d) Volume convergence of $\langle\nu\rangle(\lambda)$ for $V=10^{4},10^{6},10^{10}$, with $\lambda_{1/2}=2.455,2.591,2.605$.}
    \label{fig:fig-5}
\end{figure}

\section{Concluding Remarks}
\label{sec:Conclusions}
The three contributions stated in the Introduction and the limitations of the study are summarised as follows.

\begin{enumerate}[label=(\arabic*)]

\item \textbf{Least constraint for complex vector bundles.} The least constraint theorem for stochastic vector bundles \cite{Zhong2025LeastConstraint} is extended from real to complex fibres (Theorem~\ref{thm:least-constraint-complex}). The probability measure is a fibre coordinate of the jet bundle, and the contact potential $\mathcal{H}$ is its rate of change along the state direction; under the three conditions Eq.~\eqref{eq:three-conditions} the constraint $\varepsilon_{\mathbb{C}}$ and the contact form $\Theta_{\mathbb{C}}$ fix a unique characteristic field that preserves $\Theta_{\mathbb{C}}$ exactly and carries $\varepsilon_{\mathbb{C}}$ as a first integral. Its integral curves are the extremal paths of the action, whose stationarity is the contact Hamilton--Jacobi equation. Dissipation thus enters the geometry: the evolution is selected by the extremals of a constraint function, and the distinguished scalar is the contact potential, conserved for autonomous systems, in place of a conserved energy.

\item \textbf{The probability measure.} Parallel transport for the contact connection on the quotient line bundle $Q$ rescales a compatible complex structure $\mathcal{J}_{0}$ by the action. The resulting coordinate volume is $e^{\,n\mathcal{S}}\dd^{n}x\wedge\dd^{n}y$ (Appendix~\ref{app:proof-measure}, Steps 1--3), and disintegrating its characteristic pullback over the fixed-time on-shell graph gives the density $\mathcal{P}[W]\propto e^{\,n\mathcal{S}[W]}$ (Theorem~\ref{thm:measure}). The measure is therefore part of the kinematics of the bundle, and its weight is an exponential of the action where the standard contact formulation gives a power of the Hamiltonian.

\item \textbf{Statistics of the 2D CGLE.} A canonical transformation at fixed total intensity turns the CHJ equation into one for a real amplitude, whose travelling-wave family is solved in closed form as the reciprocal of a Jacobi-elliptic profile, Eq.~\eqref{eq:Phi-profile}, the localised profile Eq.~\eqref{Phi-soliton} being its endpoint. The on-shell action then gives the density $\mathcal{P}[W]\propto|W_{e}|^{-2}$ (Theorem~\ref{thm:reciprocal-density}), so the statistics are read off without solving the CGLE: an intensity law $I^{-3/2}$ with non-extensive index $q=5/3$ and a variance that diverges as the saturation cut-off is removed, structure-function exponents $(p-1)/2$, and an on-shell bifurcation at the analytic threshold $\lambda_{c}$ with square-root onset. The equilibrium conversion is a smooth crossover; among the three ordered scales, $\lambda_c$ is fixed by the medium, while $\lambda_{1/2}$ and $\lambda_{\mathrm{mp}}$ depend on the volume only through the infrared cap $\nu_{\max}$.

\item \textbf{Limitations of this study.} The closed-form density and the resulting statistics rely on the travelling-wave on-shell reduction of the contact Hamilton--Jacobi equation. Future extensions include more general on-shell configurations, two-point measures for spatial correlations, and vector-valued fields.
\end{enumerate}

\section*{Acknowledgments}
The authors thank the Department of Hydraulic Engineering at Tsinghua University for providing the research environment that made this work possible.

\section*{Funding}
This study was supported financially by the State Key Laboratory of Hydroscience and Engineering of Tsinghua University under Grant No. sklhse-TD-2024-F01.

\section*{Conflict of Interest}
The authors declare no conflict of interest.

\section*{Data Availability}
This is a theoretical study involving no experimental data.

\section*{Declaration}

During the preparation of this work the authors used AI-assisted tools to perform spelling and grammar checks of the manuscript, employed an AI agent to drive Wolfram Mathematica through the Model Context Protocol (MCP) to independently verify the mathematical derivations presented throughout the paper. After using these tools, the authors reviewed and edited the content as needed and take full responsibility for the content.

\clearpage
\appendix

\section{Proof of Theorem~\ref{thm:least-constraint-complex}}
\label{app:proof-thm-complex}

\begin{proof}
Throughout, $n$ is the complex fibre dimension; the proof has five steps.

\smallskip
\noindent\textbf{Step 1: contact structure and non-degeneracy.}  On the stochastic jet bundle the contact one-form is written in the jet coordinates $(P,W^{j},\Pi^{*}_{j})$ as
\begin{equation}
\Theta_{\mathbb{C}}=\dd P-\Pi^{*}_{j}\,\dd W^{j},
\qquad
\dd\Theta_{\mathbb{C}}=-\dd\Pi^{*}_{j}\wedge\dd W^{j},
\label{eq:theta-complex-proof}
\end{equation}
where $P$ is the probability measure and $\Pi^{*}_{j}$ the conjugate field of Eq.~\eqref{eq:conjugate-field-complex}. The differentials $\dd P,\dd W^{j},\dd\Pi^{*}_{j}$ are linearly independent by the injectivity condition on the coefficients $B^{\mu_{1}\cdots\mu_{k}}_{j}$.

The contact volume is nowhere vanishing: $(\dd\Theta_{\mathbb{C}})^{\wedge n}$ already contains the factor $\dd W^{1}\wedge\cdots\wedge\dd W^{n}$, which annihilates the term $\Pi^{*}_{j}\dd W^{j}$ of $\Theta_{\mathbb{C}}$, so that
\begin{equation}
\Theta_{\mathbb{C}}\wedge(\dd\Theta_{\mathbb{C}})^{\wedge n}
=(-1)^{n}n!\,\dd P\wedge\bigwedge_{j=1}^{n}\bigl(\dd\Pi^{*}_{j}\wedge\dd W^{j}\bigr)\neq0 .
\label{eq:contact-volume-general}
\end{equation}
The contact distribution is maximally non-integrable. A vector $X=a\,\partial_{P}+b^{j}\,\partial_{W^{j}}+c_{j}\,\partial_{\Pi^{*}_{j}}$ lies in $\ker\Theta_{\mathbb{C}}$ exactly when $\dd P(X)=\Pi^{*}_{j}\dd W^{j}(X)$, that is, $a=\Pi^{*}_{j}b^{j}$; on such a vector
\begin{equation}
\iota_{X}\dd\Theta_{\mathbb{C}}=-c_{j}\,\dd W^{j}+b^{j}\,\dd\Pi^{*}_{j} .
\label{eq:iota-dtheta}
\end{equation}
If $\iota_{X}\dd\Theta_{\mathbb{C}}=0$, then $c_{j}=b^{j}=0$ for every $j$, hence $a=0$ and $X=0$. Therefore $\ker\Theta_{\mathbb{C}}\cap\ker\dd\Theta_{\mathbb{C}}=\{0\}$, and $\Theta_{\mathbb{C}}$ defines a contact structure on $J^{\infty}(E_{\mathbb{C}},\mathbb{P})$.

\smallskip
\noindent\textbf{Step 2: existence of the characteristic vector field.}  On the extended phase space the contact form reads $\Theta_{\mathbb{C}}=\mathcal{H}\,\dd t-\Pi^{*}_{j}\,\dd W^{j}$ Eq.~\eqref{eq:contact-form-complex}. A vector field tangent to $\ker\dd\Theta_{\mathbb{C}}$ has the general form
\begin{equation}
X_{\mathcal{H}}=A\,\frac{\partial}{\partial t}+B^{j}\frac{\partial}{\partial W^{j}}+C_{j}\frac{\partial}{\partial\Pi^{*}_{j}},
\label{eq:X-ansatz-proof}
\end{equation}
with $A,B^{j},C_{j}$ to be determined. Writing $\dd\Theta_{\mathbb{C}}=\dd\mathcal{H}\wedge\dd t-\dd\Pi^{*}_{j}\wedge\dd W^{j}$ and expanding $\dd\mathcal{H}$ in the coordinate basis gives
\begin{equation}
\dd\Theta_{\mathbb{C}}
=-\frac{\partial\mathcal{H}}{\partial W^{j}}\,\dd t\wedge\dd W^{j}
-\frac{\partial\mathcal{H}}{\partial\Pi^{*}_{j}}\,\dd t\wedge\dd\Pi^{*}_{j}
-\dd\Pi^{*}_{j}\wedge\dd W^{j} .
\label{eq:dtheta-expansion-proof}
\end{equation}
Contracting with $X_{\mathcal{H}}$, using $\iota_{X}(\dd t\wedge\alpha)=A\,\alpha-\alpha(X)\,\dd t$ and $\iota_{X}(\dd\Pi^{*}_{j}\wedge\dd W^{j})=C_{j}\,\dd W^{j}-B^{j}\,\dd\Pi^{*}_{j}$, and collecting the coefficients of the independent one-forms $\dd t,\dd W^{j},\dd\Pi^{*}_{j}$ yields
\begin{equation}
\iota_{X_{\mathcal{H}}}\dd\Theta_{\mathbb{C}}
=\left(-A\frac{\partial\mathcal{H}}{\partial W^{j}}-C_{j}\right)\dd W^{j}
+\left(B^{j}-A\frac{\partial\mathcal{H}}{\partial\Pi^{*}_{j}}\right)\dd\Pi^{*}_{j}
+\left(B^{j}\frac{\partial\mathcal{H}}{\partial W^{j}}+C_{j}\frac{\partial\mathcal{H}}{\partial\Pi^{*}_{j}}\right)\dd t .
\label{eq:iota-dtheta-vec}
\end{equation}
The condition $\iota_{X_{\mathcal{H}}}\dd\Theta_{\mathbb{C}}=0$ therefore gives
\begin{equation}
B^{j}=A\,\frac{\partial\mathcal{H}}{\partial\Pi^{*}_{j}},
\qquad
C_{j}=-A\,\frac{\partial\mathcal{H}}{\partial W^{j}},
\label{eq:coeff-solution-proof}
\end{equation}
the coefficient of $\dd t$ then vanishing identically. The normalisation $\iota_{X_{\mathcal{H}}}\dd t=A=1$ fixes $A=1$, whence
\begin{equation}
X_{\mathcal{H}}=\frac{\partial}{\partial t}+\frac{\partial\mathcal{H}}{\partial\Pi^{*}_{j}}\frac{\partial}{\partial W^{j}}-\frac{\partial\mathcal{H}}{\partial W^{j}}\frac{\partial}{\partial\Pi^{*}_{j}},
\label{eq:char-field-solution}
\end{equation}
which is Eq.~\eqref{eq:vector}, and its flow satisfies the dynamical equations, Eqs.~\eqref{eq:contact-dynamics}.

Substituting $\varepsilon_{\mathbb{C}}=\Pi^{*}_{j}\dot{W}^{j}-\mathcal{H}$ and \eqref{eq:char-field-solution} into $\iota_{X_{\mathcal{H}}}\Theta_{\mathbb{C}}$ gives $-\varepsilon_{\mathbb{C}}$, so the first condition of Eq.~\eqref{eq:three-conditions} holds.

Cartan's formula with the first two conditions gives, on $\mathcal{E}_{\mathbb{C}}$,
\begin{equation}
L_{X_{\mathcal{H}}}\Theta_{\mathbb{C}}
=\iota_{X_{\mathcal{H}}}\dd\Theta_{\mathbb{C}}
+\dd\bigl(\iota_{X_{\mathcal{H}}}\Theta_{\mathbb{C}}\bigr)
=-\dd\varepsilon_{\mathbb{C}}=0.
\label{eq:Lie-theta-complex}
\end{equation}
Since $\varepsilon_\mathbb{C}$ is defined on $\gamma$, an integral curve of $X_{\mathcal{H}}$, i.e., $\dot\gamma(t)=X_{\mathcal{H}}\big|_{\gamma(t)}$, and 
\begin{equation}
 \dd(\varepsilon_{\mathbb{C}}\circ \gamma)
=\gamma^* \dd\varepsilon_{\mathbb{C}},
\label{eq:Lie-theta-contracted}
\end{equation}
then
\begin{equation}
(\gamma^* \dd \varepsilon_{\mathbb{C}})(\tfrac{\dd}{\dd t})
=\dd\varepsilon_{\mathbb{C}}\bigl(X_{\mathcal{H}}\bigr)
=\frac{\dd}{\dd t}\,\varepsilon_{\mathbb{C}}\bigl(\gamma(t)\bigr)=0 ,
\label{eq:eps-along-curve-complex}
\end{equation}
so $\varepsilon_{\mathbb{C}}$ is a first integral of the contact flow, its value fixed by the initial data.

\smallskip
\noindent\textbf{Step 3: uniqueness.} Let $X_{1}$ and $X_{2}$ both satisfy Eq.~\eqref{eq:three-conditions} and set $X=X_{1}-X_{2}$. Since $\ker\dd\Theta_{\mathbb{C}}$ is a linear subspace, $X\in\ker\dd\Theta_{\mathbb{C}}$. From the first condition,
\begin{equation}
\iota_{X}\Theta_{\mathbb{C}}=\iota_{X_{1}}\Theta_{\mathbb{C}}-\iota_{X_{2}}\Theta_{\mathbb{C}}=-\varepsilon_{\mathbb{C}}-(-\varepsilon_{\mathbb{C}})=0,
\label{eq:iotaX-diff}
\end{equation}
so $X\in\ker\Theta_{\mathbb{C}}$ as well. Step 1 establishes $\ker\Theta_{\mathbb{C}}\cap\ker\dd\Theta_{\mathbb{C}}=\{0\}$; hence $X=0$, that is, $X_{1}=X_{2}$, and the vector field is uniquely determined by $\Theta_{\mathbb{C}}$ and $\varepsilon_{\mathbb{C}}$.

\smallskip
\noindent\textbf{Step 4: Poisson bracket and dynamical equations.}  On functions of $(W,\Pi^{*})$ define the bracket
\begin{equation}
\{F,G\}=\frac{\partial F}{\partial W^{j}}\frac{\partial G}{\partial\Pi^{*}_{j}}-\frac{\partial F}{\partial\Pi^{*}_{j}}\frac{\partial G}{\partial W^{j}} .
\label{eq:poisson-bracket-proof}
\end{equation}
Substitution gives the fundamental brackets
\begin{equation}
\{W^{i},W^{j}\}=0,\qquad
\{\Pi^{*}_{i},\Pi^{*}_{j}\}=0,\qquad
\{W^{i},\Pi^{*}_{j}\}=\delta^{i}_{j}.
\label{eq:fundamental-proof}
\end{equation}
With this bracket Eq.~\eqref{eq:contact-dynamics} reads
\begin{equation}
\frac{\partial W^{j}}{\partial t}=\{W^{j},\mathcal{H}\},\qquad
\frac{\partial\Pi^{*}_{j}}{\partial t}=\{\Pi^{*}_{j},\mathcal{H}\} .
\end{equation}
For the constraint $\varepsilon_{\mathbb{C}}=\Pi^{*}_{j}\dot{W}^{j}-\mathcal{H}$, using $\{\mathcal{H},\mathcal{H}\}=0$,
\begin{equation}
\frac{\partial\varepsilon_{\mathbb{C}}}{\partial t}
=-\{\Pi^{*}_{j}\dot{W}^{j},\mathcal{H}\}
=-\{\varepsilon_{\mathbb{C}}+\mathcal{H},\mathcal{H}\}
=-\{\varepsilon_{\mathbb{C}},\mathcal{H}\},
\label{eq:eps-bracket-proof}
\end{equation}
and any function $F$ evolves as $\dd F/\dd t=\partial_{t}F+\{F,\mathcal{H}\}$, which is Eq.~\eqref{eq:dynamics-F}.

\smallskip
\noindent\textbf{Step 5: extremal paths and the contact Hamilton--Jacobi equation.} Define the action along a characteristic line by
\begin{equation}
\mathcal{S}=\int\varepsilon_{\mathbb{C}}\,\dd t=-\int\Theta_{\mathbb{C}},
\label{eq:action-complex-proof}
\end{equation}
using $\varepsilon_{\mathbb{C}}\,\dd t=\Pi^{*}_{j}\dd W^{j}-\mathcal{H}\,\dd t=-\Theta_{\mathbb{C}}$. Varying $\mathcal{S}$ with the endpoints held fixed, and integrating the term $\Pi^{*}_{j}\delta\dot{W}^{j}$ by parts with vanishing boundary contribution,
\begin{equation}
\delta\mathcal{S}
=\int\left[\left(\dot{W}^{j}-\frac{\partial\mathcal{H}}{\partial\Pi^{*}_{j}}\right)\delta\Pi^{*}_{j}
-\left(\dot{\Pi}^{*}_{j}+\frac{\partial\mathcal{H}}{\partial W^{j}}\right)\delta W^{j}\right]\dd t .
\label{eq:variation-proof}
\end{equation}
The stationarity $\delta\mathcal{S}=0$ for arbitrary variations $\delta W^{j},\delta\Pi^{*}_{j}$ is therefore equivalent to the dynamical equations, Eqs.~\eqref{eq:contact-dynamics}, so the integral curves of $X_{\mathcal{H}}$ are the extremal paths of $\mathcal{S}$. Writing the second equality in Eq.~\eqref{eq:action-complex-proof} in coordinates and comparing coefficients gives
\begin{equation}
\frac{\partial\mathcal{S}}{\partial W^{j}}=\Pi^{*}_{j},
\qquad
\frac{\partial\mathcal{S}}{\partial t}=-\mathcal{H},
\qquad
\frac{\partial\mathcal{S}}{\partial\Pi^{*}_{j}}=0 .
\label{eq:canonical-relation-proof}
\end{equation}
The last relation states that $\mathcal{S}$ descends to a function of the configuration variables. Substituting the first relation into the second gives the contact Hamilton--Jacobi equation
\begin{equation}
\frac{\partial\mathcal{S}}{\partial t}+\mathcal{H}\!\left[W^{j},\frac{\partial\mathcal{S}}{\partial W^{j}},t\right]=0 .
\label{eq:chj-proof}
\end{equation}
\end{proof}

\section{Proof of Lemma~\ref{lem:contact-connection}}
\label{app:proof-contact-connection}

\begin{proof}
\smallskip
\noindent\textbf{Step 1: The quotient line bundle.}
By statement (1) of Theorem~\ref{thm:least-constraint-complex} the contact volume Eq.~\eqref{eq:contact-volume-complex} is non-degenerate, so $\Theta_{\mathbb{C}}$ vanishes nowhere: a zero of $\Theta_{\mathbb{C}}$ would force a zero of $\Theta_{\mathbb{C}}\wedge(\dd\Theta_{\mathbb{C}})^{\wedge n}$. Hence $\ker\Theta_{\mathbb{C}}$ is a complex subbundle of $T\mathcal{E}_{\mathbb{C}}$ of rank $2n$, and the quotient Eq.~\eqref{eq:contact-line-bundle} is a complex line bundle. The map $\bar{\Theta}_{\mathbb{C}}\colon Q\to\mathcal{E}_{\mathbb{C}}\times\mathbb{C}$ induced by $\Theta_{\mathbb{C}}$ is fibrewise a nonzero linear map from a one-dimensional space to $\mathbb{C}$, hence a fibrewise isomorphism, and sections of $Q$ are identified with smooth complex functions on $\mathcal{E}_{\mathbb{C}}$.

\smallskip
\noindent\textbf{Step 2: $\nabla^{Q}$ is a connection.}
Both $X(\sigma)$ and $\Theta_{\mathbb{C}}(X)$ are $C^{\infty}$-linear in $X$, so for smooth functions $f,g$ and vector fields $X,Y$,
\begin{equation}
\nabla^{Q}_{fX+gY}\,\sigma=f\,\nabla^{Q}_{X}\,\sigma+g\,\nabla^{Q}_{Y}\,\sigma .
\label{eq:nablaQ-Cinfinity}
\end{equation}
For a smooth function $f$, the Leibniz rule $X(f\sigma)=X(f)\sigma+f\,X(\sigma)$ gives
\begin{equation}
\nabla^{Q}_{X}(f\,\sigma)=X(f)\,\sigma+f\,\nabla^{Q}_{X}\,\sigma .
\label{eq:nablaQ-leibniz}
\end{equation}
Eqs.~\eqref{eq:nablaQ-Cinfinity} and \eqref{eq:nablaQ-leibniz} are the two axioms of a connection.
\end{proof}

\section{Proof of Theorem~\ref{thm:measure}}\label{app:proof-measure}

\begin{proof}
Let $\gamma$ be an integral curve of $X_{\mathcal{H}}$, let $\omega_{\mathbb{C}}=\dd\Theta_{\mathbb{C}}\vert_{\ker\Theta_{\mathbb{C}}}$ be the symplectic form whose rank fixes $n$ in Eq.~\eqref{eq:contact-volume-complex}, and let $\mathcal{S}=-\int\Theta_{\mathbb{C}}$ (Eq.~\eqref{eq:action-complex}) be the action, normalised by $\mathcal{S}(\gamma(0))=0$. Let $\{e_{i}\}_{i=1}^{n}$ be a complex Lagrangian frame in $\ker\Theta_{\mathbb{C}}$ along $\gamma$, so that $\omega_{\mathbb{C}}(e_{i},e_{j})=0$. Let $\mathcal{J}_{0}$ be an $\omega_{\mathbb{C}}$-compatible complex structure on $\ker\Theta_{\mathbb{C}}$ at $\gamma(0)$ with $\omega_{\mathbb{C}}(e_{i},\mathcal{J}_{0}e_{j})=\delta_{ij}$. Using the trivialisation $\bar{\Theta}_{\mathbb{C}}:Q\simeq\mathcal{E}_{\mathbb{C}}\times\mathbb{C}$ of Lemma~\ref{lem:contact-connection}, put $q_{0}:=\bar{\Theta}_{\mathbb{C}}^{-1}(1)$ and identify $A\otimes q_{0}\in\mathrm{End}(\ker\Theta_{\mathbb{C}})\otimes Q$ with $A\in\mathrm{End}(\ker\Theta_{\mathbb{C}})$. Let $\mathcal{J}=A\otimes q_{0}$ be the parallel transport of $\mathcal{J}_{0}$ along $\gamma$ for the tensor product of the flow connection on $\mathrm{End}(\ker\Theta_{\mathbb{C}}\vert_{\gamma})$ with $\nabla^{Q}$,
\begin{equation}
\nabla^{Q}_{X_{\mathcal{H}}}\mathcal{J}=0,
\qquad
\mathcal{J}\big|_{\gamma(0)}=\mathcal{J}_{0}.
\label{eq:J-parallel}
\end{equation}
Write $g_{\mathcal{J}'}(U,V):=\omega_{\mathbb{C}}(U,\mathcal{J}'V)$ for an endomorphism $\mathcal{J}'$ of $\ker\Theta_{\mathbb{C}}$. Finally, let $(\dd x^{i},\dd y^{i})_{i=1}^{n}$ be a complex coframe of $\ker\Theta_{\mathbb{C}}$ along $\gamma$, dual to a $g_{\mathcal{J}}$-orthonormal frame.

Throughout, the linear algebra is over $\mathbb{C}$ and orthonormal means $g_{\mathcal{J}_{0}}(e_{i},e_{j})=\delta_{ij}$. The flow $\phi_{t}$ of $X_{\mathcal{H}}$ satisfies $\gamma(t)=\phi_{t}(\gamma(0))$, with differential $\phi_{t*}$. The $\omega_{\mathbb{C}}$-compatibility of $\mathcal{J}_{0}$ reads
\begin{equation}
\mathcal{J}_{0}^{2}=-\mathrm{id},
\qquad
g_{\mathcal{J}_{0}}\ \ \text{symmetric and non-degenerate},
\label{eq:J0-compatible}
\end{equation}
and the frame satisfies
\begin{equation}
e_{i}(t)=\phi_{t*}\,e_{i}(0).
\label{eq:transport}
\end{equation}

\smallskip
\noindent\textbf{Step 1: The parallel section.}
By the third condition of Eq.~\eqref{eq:three-conditions}, $L_{X_{\mathcal{H}}}\Theta_{\mathbb{C}}=0$, so the flow preserves $\ker\Theta_{\mathbb{C}}$ and the flow trivialisation is well defined.
The rate of change of the action along $\gamma$ is
\begin{equation}
L_{X_{\mathcal{H}}}\mathcal{S}
=\frac{\dd}{\dd t}\,\mathcal{S}\bigl(\gamma(t)\bigr)
=-\Theta_{\mathbb{C}}\bigl(X_{\mathcal{H}}\bigr)\big|_{\gamma}
=\varepsilon_{\mathbb{C}},
\label{eq:dS-dt-eps}
\end{equation}
by Eq.~\eqref{eq:action-complex} and the first condition of Eq.~\eqref{eq:three-conditions}.
In the flow trivialisation $\mathcal{J}_{0}$ is constant and
\begin{equation}
\nabla^{Q}_{X_{\mathcal{H}}}\mathcal{J}
=L_{X_{\mathcal{H}}}\mathcal{J}+\varepsilon_{\mathbb{C}}\,\mathcal{J}
=\frac{\dd}{\dd t}\,\mathcal{J}+\varepsilon_{\mathbb{C}}\,\mathcal{J},
\label{eq:nablaQ-J-flow}
\end{equation}
so the parallel condition Eq.~\eqref{eq:J-parallel} is
\begin{equation}
\frac{\dd}{\dd t}\,\mathcal{J}+\varepsilon_{\mathbb{C}}\,\mathcal{J}=0,
\qquad
\mathcal{J}\big|_{\gamma(0)}=\mathcal{J}_{0},
\label{eq:J-parallel-trivial}
\end{equation}
which, by Eq.~\eqref{eq:dS-dt-eps} and the normalisation $\mathcal{S}(\gamma(0))=0$, has the unique solution
\begin{equation}
\mathcal{J}=e^{-\mathcal{S}}\,\mathcal{J}_{0}
\qquad\text{along }\gamma .
\label{eq:J-conformal}
\end{equation}

\smallskip
\noindent\textbf{Step 2: The symplectic frame.}
Since $L_{X_{\mathcal{H}}}\omega_{\mathbb{C}}=L_{X_{\mathcal{H}}}\dd\Theta_{\mathbb{C}}=\dd(L_{X_{\mathcal{H}}}\Theta_{\mathbb{C}})=0$, the flow preserves $\omega_{\mathbb{C}}$, and for $U,V\in\ker\Theta_{\mathbb{C}}$,
\begin{equation}
\omega_{\mathbb{C}}\bigl(\phi_{t*}U,\phi_{t*}V\bigr)
=\omega_{\mathbb{C}}(U,V).
\label{eq:omega-flow-invariance}
\end{equation}
Since $\mathcal{J}_{0}$ and the frame are constant in the flow trivialisation, the $\mathcal{J}_{0}$-orthonormality imposed at $\gamma(0)$ holds along all of $\gamma$:
\begin{equation}
g_{\mathcal{J}_{0}}\bigl(e_{i},e_{j}\bigr)
=\omega_{\mathbb{C}}\bigl(e_{i},\mathcal{J}_{0}e_{j}\bigr)
=\delta_{ij}.
\label{eq:frame-orthonormal-delta}
\end{equation}
Let $f_{i}:=\mathcal{J}_{0}e_{i}$. By the symmetry of $g_{\mathcal{J}_{0}}$ in Eq.~\eqref{eq:J0-compatible},
\begin{equation}
\begin{gathered}
\omega_{\mathbb{C}}(U,V)=-g_{\mathcal{J}_{0}}(U,\mathcal{J}_{0}V)=g_{\mathcal{J}_{0}}(\mathcal{J}_{0}U,V),\\
\omega_{\mathbb{C}}(\mathcal{J}_{0}U,\mathcal{J}_{0}V)=-g_{\mathcal{J}_{0}}(\mathcal{J}_{0}U,\mathcal{J}_{0}^{2}V)=\omega_{\mathbb{C}}(U,V),
\end{gathered}
\label{eq:J0-symplectic}
\end{equation}
so $\{e_{i},f_{i}\}$ is symplectic: $\omega_{\mathbb{C}}(e_{i},f_{j})=\delta_{ij}$, $\omega_{\mathbb{C}}(e_{i},e_{j})=0$, and $\omega_{\mathbb{C}}(f_{i},f_{j})=0$. In the dual basis $\{e^{i},f^{i}\}$,
\begin{equation}
\begin{gathered}
\omega_{\mathbb{C}}=\sum_{i}e^{i}\wedge f^{i},\\
\omega_{\mathbb{C}}^{\wedge n}=n!\,e^{1}\wedge f^{1}\wedge\cdots\wedge e^{n}\wedge f^{n},
\end{gathered}
\label{eq:omega-canonical}
\end{equation}
\smallskip
\noindent\textbf{Step 3: Coordinate density.}
By Eq.~\eqref{eq:J-conformal}, the metric $g_{\mathcal{J}}$ is conformal to $g_{\mathcal{J}_{0}}$,
\begin{equation}
g_{\mathcal{J}}(U,V)=\omega_{\mathbb{C}}(U,\mathcal{J}V)=e^{-\mathcal{S}}\,g_{\mathcal{J}_{0}}(U,V).
\label{eq:gJ-scaled}
\end{equation}
Let $(\dd x_{0}^{i},\dd y_{0}^{i})_{i=1}^{n}$ be the coframe dual to $\{e_{i},f_{i}\}$ at a point of $\gamma$. By Eq.~\eqref{eq:omega-canonical},
\begin{equation}
\dd\mu=\frac{\omega_{\mathbb{C}}^{\wedge n}}{n!}
=\dd x_{0}^{1}\wedge\dd y_{0}^{1}\wedge\cdots\wedge\dd x_{0}^{n}\wedge\dd y_{0}^{n} .
\label{eq:mu-orthonormal}
\end{equation}
A $g_{\mathcal{J}}$-orthonormal frame is $\tilde{e}_{i}=e^{\mathcal{S}/2}e_{i}$, $\tilde{f}_{i}=e^{\mathcal{S}/2}f_{i}$; writing $v=x_{0}^{i}e_{i}+y_{0}^{i}f_{i}=x^{i}\tilde{e}_{i}+y^{i}\tilde{f}_{i}$,
\begin{equation}
\begin{gathered}
x^{i}=e^{-\mathcal{S}/2}\,x_{0}^{i},
\qquad
y^{i}=e^{-\mathcal{S}/2}\,y_{0}^{i},\\
\dd x^{1}\wedge\dd y^{1}\wedge\cdots\wedge\dd x^{n}\wedge\dd y^{n}
=e^{-n\mathcal{S}}\,\dd x_{0}^{1}\wedge\dd y_{0}^{1}\wedge\cdots\wedge\dd x_{0}^{n}\wedge\dd y_{0}^{n}.
\end{gathered}
\label{eq:W-jacobian}
\end{equation}
Combining Eq.~\eqref{eq:mu-orthonormal} and Eq.~\eqref{eq:W-jacobian},
\begin{equation}
\dd\mu=\frac{\omega_{\mathbb{C}}^{\wedge n}}{n!}
=e^{\,n\mathcal{S}}\,\dd x^{1}\wedge\dd y^{1}\wedge\cdots\wedge\dd x^{n}\wedge\dd y^{n}.
\label{eq:mu-coordinate}
\end{equation}

\smallskip
\noindent\textbf{Step 4: Characteristic projection and the probability measure on $E_{\mathbb{C}}$.}

\smallskip
\noindent\emph{4.1 The on-shell graph.}
For fixed $t$, let $\varpi$ denote the projection to the configuration
variables and define the on-shell graph
\begin{equation}
\begin{gathered}
\mathcal{L}_{\mathcal{S},t}
:=\left\{(W,\Pi^{*}):\;
\Pi^{*}_{j}=\frac{\partial\mathcal{S}}{\partial W^{j}}(W,t)
\right\},\\
\sigma_{t}:E_{\mathbb{C}}\to\mathcal{L}_{\mathcal{S},t},\\
\sigma_{t}(W)=\left(W,\frac{\partial\mathcal{S}}{\partial W}(W,t)\right).
\end{gathered}
\label{eq:on-shell-section}
\end{equation}
Thus $\varpi\circ\sigma_{t}=\mathrm{id}_{E_{\mathbb{C}}}$, so that
$\mathcal{L}_{\mathcal{S},t}$ is diffeomorphic to $E_{\mathbb{C}}$. In a
local reduction chart transverse to the graph, write it as
$y^{i}=Y^{i}(x,t)$, where $x$ denotes the retained configuration
coordinates and $y$ the transverse coordinates of
Eq.~\eqref{eq:mu-coordinate}.

\smallskip
\noindent\emph{4.2 Characteristic projection and the ambient pullback.}
Put $K:=\ker\Theta_{\mathbb{C}}$. The first two conditions of
Eq.~\eqref{eq:three-conditions} give
\begin{equation}
\begin{gathered}
\Theta_{\mathbb{C}}(X_{\mathcal{H}})
=-\varepsilon_{\mathbb{C}}\neq0,
\qquad
\iota_{X_{\mathcal{H}}}\dd\Theta_{\mathbb{C}}=0,\\
K\cap\langle X_{\mathcal{H}}\rangle=\{0\},
\qquad
T\mathcal{E}_{\mathbb{C}}
=K\oplus\langle X_{\mathcal{H}}\rangle .
\end{gathered}
\label{eq:tangent-splitting}
\end{equation}
For each $v\in T\mathcal{E}_{\mathbb{C}}$, write the proposed projection in
the form $\pi_{K}v=v+a(v)X_{\mathcal{H}}$ and determine the scalar $a(v)$ by
requiring $\pi_{K}v\in K$:
\begin{equation}
\begin{gathered}
0=\Theta_{\mathbb{C}}(\pi_{K}v)
=\Theta_{\mathbb{C}}(v)+a(v)\Theta_{\mathbb{C}}(X_{\mathcal{H}})
=\Theta_{\mathbb{C}}(v)-\varepsilon_{\mathbb{C}}a(v),\\
a(v)=\frac{\Theta_{\mathbb{C}}(v)}{\varepsilon_{\mathbb{C}}},
\qquad
\pi_{K}v=v+
\frac{\Theta_{\mathbb{C}}(v)}{\varepsilon_{\mathbb{C}}}X_{\mathcal{H}} .
\end{gathered}
\label{eq:piK-definition}
\end{equation}
Since $\Theta_{\mathbb{C}}$ is linear, $\pi_{K}$ is a pointwise linear
projection onto $K$.

The volume $\dd\mu$ is a local section of $\Lambda^{2n}K^{*}$.
Since $\pi_K:T\mathcal{E}_{\mathbb{C}}\to K$ is a bundle projection
over the identity, its pointwise algebraic pullback defines
\begin{equation}
\begin{gathered}
\widehat{\dd\mu}:=\pi_{K}^{*}\dd\mu
\in\Gamma_{\mathrm{loc}}\!\left(
\Lambda^{2n}_{\mathbb{C}}T^*\mathcal{E}_{\mathbb{C}}
\right),\\
\widehat{\dd\mu}_{z}(v_{1},\ldots,v_{2n})
:=\dd\mu_{z}(\pi_{K}v_{1},\ldots,\pi_{K}v_{2n}),
\qquad
z\in\mathcal{E}_{\mathbb{C}},\quad
v_{i}\in T_{z}\mathcal{E}_{\mathbb{C}} .
\end{gathered}
\label{eq:mu-hat-pullback}
\end{equation}
For each $v_{i}\in T\mathcal{E}_{\mathbb{C}}$, the definition of $\pi_{K}$ gives the decomposition
\begin{equation}
v_{i}=\pi_{K}v_{i}
-\frac{\Theta_{\mathbb{C}}(v_{i})}{\varepsilon_{\mathbb{C}}}X_{\mathcal{H}},
\qquad i=1,\ldots,2n .
\label{eq:extension-decomposition}
\end{equation}

\smallskip
\noindent\emph{4.3 Identification with the ambient symplectic volume.}
For $v,w\in T\mathcal{E}_{\mathbb{C}}$, Eq.~\eqref{eq:extension-decomposition} and $\iota_{X_{\mathcal{H}}}\dd\Theta_{\mathbb{C}}=0$ annihilate every term of the bilinear expansion that contains $X_{\mathcal{H}}$, so
\begin{equation}
\dd\Theta_{\mathbb{C}}(v,w)
=\dd\Theta_{\mathbb{C}}(\pi_{K}v,\pi_{K}w)
=\omega_{\mathbb{C}}(\pi_{K}v,\pi_{K}w)
=(\pi_{K}^{*}\omega_{\mathbb{C}})(v,w) ,
\label{eq:dTheta-pullback-evaluation}
\end{equation}
the second equality being the restriction
$\dd\Theta_{\mathbb{C}}|_{K}=\omega_{\mathbb{C}}$ and the third the
definition of the algebraic pullback. Since $v,w$ are arbitrary,
\begin{equation}
\dd\Theta_{\mathbb{C}}=\pi_{K}^{*}\omega_{\mathbb{C}} .
\label{eq:dTheta-pullback}
\end{equation}
Algebraic pullback commutes with the wedge product, so
Eqs.~\eqref{eq:mu-hat-pullback} and \eqref{eq:dTheta-pullback} give
\begin{equation}
\widehat{\dd\mu}
=\pi_{K}^{*}\dd\mu
=\pi_{K}^{*}\!\left(\frac{\omega_{\mathbb{C}}^{\wedge n}}{n!}\right)
=\frac{(\pi_{K}^{*}\omega_{\mathbb{C}})^{\wedge n}}{n!}
=\frac{(\dd\Theta_{\mathbb{C}})^{\wedge n}}{n!}.
\label{eq:mu-hat-def}
\end{equation}
Extend the coframe of Eq.~\eqref{eq:mu-coordinate} from $K$ by algebraic pullback under $\pi_{K}$, and retain the notation $\dd x^{i}$ and $\dd y^{i}$ for the pulled-back covectors. Since $\dd t(X_{\mathcal{H}})=1$, the fixed-$t$ slice is transverse to $X_{\mathcal{H}}$, so these covectors restrict to a coframe on that slice. Applying $\pi_{K}^{*}$ to Eq.~\eqref{eq:mu-coordinate} gives
\begin{equation}
\begin{aligned}
\widehat{\dd\mu}
&=\pi_{K}^{*}\!\left(e^{\,n\mathcal{S}}\,
\dd x^{1}\wedge\dd y^{1}\wedge\cdots\wedge
\dd x^{n}\wedge\dd y^{n}\right)\\
&=e^{\,n\mathcal{S}}\,
(\pi_{K}^{*}\dd x^{1})\wedge(\pi_{K}^{*}\dd y^{1})\wedge\cdots\wedge
(\pi_{K}^{*}\dd x^{n})\wedge(\pi_{K}^{*}\dd y^{n})\\
&=e^{\,n\mathcal{S}}\,\dd x^{1}\wedge\dd y^{1}\wedge\cdots\wedge
\dd x^{n}\wedge\dd y^{n}
=:e^{\,n\mathcal{S}}\,\dd^{n}x\wedge\dd^{n}y .
\end{aligned}
\label{eq:mu-hat-coordinate}
\end{equation}

\smallskip
\noindent\emph{4.4 Disintegration over the on-shell graph.}
For fixed $t$, put
$\Sigma_t:=T^*E_{\mathbb{C}}\times\{t\}$, let
$i_t:\Sigma_t\hookrightarrow\mathcal{E}_{\mathbb{C}}$ be the inclusion,
and let $\varpi_t:\Sigma_t\to E_{\mathbb{C}}$ be the configuration
projection. The on-shell graph is regarded as a submanifold
$\mathcal{L}_{\mathcal{S},t}\subset\Sigma_t$, and its pullback
$\widehat{\dd\mu}_{t}:=i_t^*\widehat{\dd\mu}
\in\Gamma_{\mathrm{loc}}\!\left(
\Lambda^{2n}_{\mathbb{C}}T^*\Sigma_t
\right)$
is a top-degree form on the fixed-time slice. Let
$\delta_{\mathcal{L}_{\mathcal{S},t}}$ be the graph current on
$\Sigma_t$, normalised by the transverse coordinate volume of
Eq.~\eqref{eq:mu-hat-coordinate}. The reduced measure is defined by
vertical disintegration:
\begin{equation}
\int_{E_{\mathbb{C}}}f\,\dd\mu_{\mathrm{red},t}
:=\int_{\Sigma_t}(f\circ\varpi_t)\,
\delta_{\mathcal{L}_{\mathcal{S},t}}\,\widehat{\dd\mu}_{t}
\label{eq:vertical-disintegration}
\end{equation}
for every compactly supported test function $f$ on $E_{\mathbb{C}}$.

\smallskip
\noindent\emph{4.5 Transverse integration and its Jacobian.}
In the chart $y^{i}=Y^{i}(x,t)$, put
\begin{equation}
G^{i}(x,y,t):=y^{i}-Y^{i}(x,t),
\qquad
\delta_{\mathcal{L}_{\mathcal{S},t}}=\delta^{(n)}(G).
\label{eq:transverse-constraints}
\end{equation}
The Jacobian relevant to integration over the transverse variables is
\begin{equation}
\frac{\partial G^{i}}{\partial y^{j}}=\delta^{i}_{j},
\qquad
\det\!\left(\frac{\partial G}{\partial y}\right)=1,
\qquad
\int\delta^{(n)}\!\bigl(y-Y(x,t)\bigr)\,\dd^{n}y
=\left|\det\!\left(\frac{\partial G}{\partial y}\right)\right|^{-1}=1 .
\label{eq:transverse-jacobian}
\end{equation}
In real coordinate representatives $\delta^{(n)}$ and $\dd^{n}y$ are the real
$2n$-dimensional transverse distribution and volume element, so
$\partial G/\partial y$ is the real $2n\times2n$ identity. Using
Eq.~\eqref{eq:mu-hat-coordinate}, choosing the positive orientation for the
coordinate volume elements, and integrating out $y$, we obtain
\begin{equation}
\begin{aligned}
\int f\,\dd\mu_{\mathrm{red},t}
&=\int f(x)\,e^{\,n\mathcal{S}(x,y,t)}\,
\delta^{(n)}\!\bigl(y-Y(x,t)\bigr)\,\dd^{n}x\,\dd^{n}y\\
&=\int f(x)\,e^{\,n\mathcal{S}(x,Y(x,t),t)}
\left[\int\delta^{(n)}\!\bigl(y-Y(x,t)\bigr)\,\dd^{n}y\right]\dd^{n}x\\
&=\int f(x)\,e^{\,n\mathcal{S}(x,Y(x,t),t)}\,\dd^{n}x .
\end{aligned}
\label{eq:disintegration-computation}
\end{equation}
The transverse integral first defines the top form $\dd\nu_t$ on
$\mathcal{L}_{\mathcal{S},t}$. Its pullback by $\sigma_t$ gives the
measure on $E_{\mathbb{C}}$:
\begin{equation}
\begin{gathered}
\dd\nu_t:=e^{\,n\mathcal{S}(x,Y(x,t),t)}
\left.\dd^{n}x\right|_{\mathcal{L}_{\mathcal{S},t}},\\
\dd\mu_{\mathrm{red},t}=\sigma_t^*\dd\nu_t,
\qquad
\mathcal{D}W:=\sigma_t^*\!\left(
\left.\dd^{n}x\right|_{\mathcal{L}_{\mathcal{S},t}}
\right).
\end{gathered}
\label{eq:DW-def}
\end{equation}
Since $\varpi_t\circ\sigma_t=\mathrm{id}_{E_{\mathbb{C}}}$ and
$\mathcal{S}(x,Y(x,t),t)=\mathcal{S}[W]$ on the on-shell graph,
$\mathcal{D}W$ is the positively oriented volume element in the
statement of Theorem~\ref{thm:measure}. Since $f$ is arbitrary,
Eq.~\eqref{eq:disintegration-computation} proves
\begin{equation}
\dd\mu_{\mathrm{red},t}
=e^{\,n\mathcal{S}[W]}\,\mathcal{D}W .
\label{eq:measure-reduced-value}
\end{equation}
We suppress the subscript $t$ on
$\dd\mu_{\mathrm{red}}$ when the on-shell graph is understood.

\smallskip
\noindent\emph{4.6 Normalisation and positivity.}
Write
\begin{equation}
\mathcal{Z}(t):=\int_{E_{\mathbb{C}}}
e^{\,n\mathcal{S}[W]}\,\mathcal{D}W
=\int_{E_{\mathbb{C}}}\dd\mu_{\mathrm{red},t}.
\label{eq:reduced-normalisation}
\end{equation}
By hypothesis $0<\mathcal{Z}(t)<\infty$. At the fixed time $t$, write $\mathcal{Z}=\mathcal{Z}(t)$ and define
\begin{equation}
\dd P:=\frac{1}{\mathcal{Z}}\,\dd\mu_{\mathrm{red}}
=\frac{1}{\mathcal{Z}}\,e^{\,n\mathcal{S}[W]}\,\mathcal{D}W .
\label{eq:measure-density-final}
\end{equation}
Since $\mathcal{S}[W]$ is real-valued and $\mathcal{D}W$ is positively
oriented, $e^{\,n\mathcal{S}[W]}>0$ and $\dd P$ is a positive measure. Its
total mass is
\begin{equation}
\int_{E_{\mathbb{C}}}\dd P
=\frac{1}{\mathcal{Z}}\int_{E_{\mathbb{C}}}
e^{\,n\mathcal{S}[W]}\,\mathcal{D}W
=\frac{\mathcal{Z}}{\mathcal{Z}}=1 .
\end{equation}
Hence $\dd P$ is a probability measure with density $\mathcal{P}[W]=\mathcal{Z}^{-1}e^{\,n\mathcal{S}[W]}$ with respect to $\mathcal{D}W$, which proves the theorem.

\end{proof}

\section{Derivation of the Travelling-Wave Solution}
\label{sec:Solution-CHJ}

This appendix solves the profile equation, Eq.~\eqref{eq:reduced-ode}, of
\S\ref{subsec:HJ_traveling_wave} and derives the closed forms quoted in
\S\ref{subsec:1d_solution}. Throughout,
\begin{equation}
A\equiv k_x^{2}-\mu-\omega+C\in\mathbb{R},
\qquad
B\equiv\frac{J}{2\pi}\in\mathbb{R},
\qquad B>0,
\label{eq:AB-def}
\end{equation}
so that Eq.~\eqref{eq:reduced-ode} reads
\begin{equation}\label{eq:reduced-ode-AB}
\Phi''=A\,\Phi-B\,\Phi^{3}+B^{2}\Phi^{5},
\end{equation}
where:
\begin{enumerate}[label=(\arabic*)]
    \item $\Phi(y)\in\mathbb{R}$ is the unknown profile, with normalisation
    $\int\Phi^{2}(y)\,\dd y=1$, understood per period for the periodic branch;
    \item $k_x$ is the fixed wavenumber in the $x$-direction, $\mu$ is the
    constant of Eq.~\eqref{eq:qcgle} and $C$ the constant of
    Eq.~\eqref{eq:contact-potential}; $\omega$ is the frequency of the
    travelling-wave ansatz and is free, and $J$ is the total intensity over the
    $x$-range of length $2\pi$ and the profile domain in $y$, which the
    normalisation fixes;
    \item $a\in\mathbb{R}$, introduced in Step~1, is the single integration
    constant of the reduction.
\end{enumerate}

\noindent\textbf{Step 0: Transverse Phase}.
Our primary interest lies in localised and periodic wave profiles. We now examine
whether such profiles can persist when the solution ansatz includes a transverse
phase $\psi=k_{x}x-\omega t+\varphi(y)$ with $\varphi'\neq0$. Substituting this phase
turns Eq.~\eqref{eq:reduced-ode-AB} into
\begin{equation}\label{eq:phase-general-ode}
\Phi''=A\,\Phi-B\,\Phi^{3}+B^{2}\Phi^{5}
+2c_{2}\varphi'\Phi'+\varphi'^{2}\Phi+c_{2}\varphi''\Phi ,
\end{equation}
the last three terms being the contributions of $\varphi$ to $\nabla^{2}W$ and
of $ic_{2}$ to their product.

Write $u=(\Phi')^{2}$. Differentiating $u$ with respect to $\Phi$ gives
twice $\Phi''$, so $u$ is a single-valued function of $\Phi$, and the profile
possesses the required first integral for our localised or periodic solutions,
exactly when the right-hand side of
Eq.~\eqref{eq:phase-general-ode} depends on $\Phi$ alone. That right-hand side
carries $\varphi'$ and $\varphi'^{2}+c_{2}\varphi''$ as coefficients, so this
requires $\varphi'$ to be constant.

Let $\varphi'=\kappa_{0}\neq0$, hence $\varphi''=0$, and write
$f(\Phi)\equiv(A+\kappa_{0}^{2})\Phi-B\Phi^{3}+B^{2}\Phi^{5}$. Along a
solution of Eq.~\eqref{eq:phase-general-ode},
\begin{equation}\label{eq:phase-obstruction}
\frac{\dd}{\dd y}(\Phi')^{2}
=2\Phi'\bigl[f(\Phi)+2c_{2}\kappa_{0}\Phi'\bigr]
\quad\Longrightarrow\quad
\frac{\dd u}{\dd\Phi}=2f(\Phi)+\frac{4c_{2}\kappa_{0}\,u}{\Phi'}
\ \ \text{if}\ \ u=(\Phi')^{2},
\end{equation}
whose right-hand side depends on the sign of $\Phi'$ unless $c_{2}\kappa_{0}=0$.
A single-valued first integral $u(\Phi)$ therefore requires $\kappa_{0}=0$
whenever $c_{2}\neq0$.

For $c_{2}\kappa_{0}\neq0$ the first integral is lost, and with it the
periodic and the localised profiles we seek. Multiplying
Eq.~\eqref{eq:phase-general-ode} by $\Phi'$ integrates it instead to
\begin{equation}\label{eq:phase-energy}
\frac{\dd}{\dd y}\Bigl[\tfrac12(\Phi')^{2}-\int^{\Phi}\!f(s)\,\dd s\Bigr]
=2c_{2}\kappa_{0}\,(\Phi')^{2},
\end{equation}
where the integral of $f$ is taken from the origin. The bracketed quantity takes
a single value at the two ends of a profile of period $T$, and a single value at
the two ends of a profile with $\Phi\to0$ and $\Phi'\to0$; either profile
therefore has a vanishing integral of $(\Phi')^{2}$ over its range, and a
vanishing $\Phi'$. A periodic or a localised chirped profile is therefore a
constant amplitude, and a constant amplitude $\Phi\equiv\Phi_{0}>0$, for which
$\Phi'=\Phi''=0$, satisfies Eq.~\eqref{eq:phase-general-ode} only if
\begin{equation}\label{eq:phase-planewave}
(A+\kappa_{0}^{2})\,\Phi_{0}-B\,\Phi_{0}^{3}+B^{2}\Phi_{0}^{5}=0 .
\end{equation}

For $c_{2}=0$, Eq.~\eqref{eq:phase-general-ode} is
Eq.~\eqref{eq:reduced-ode-AB} itself with $A$ replaced by $A+\kappa_{0}^{2}$:
the profile is a member of the same family, with the same $\Phi$ and the same
$|W_{e}|$, relabelled by the frequency $\omega$, which is free.

We therefore set $\kappa_{0}=0$ in Steps~1--6, and Eq.~\eqref{eq:phase-general-ode} reduces to the amplitude equation Eq.~\eqref{eq:reduced-ode-AB}.

\noindent\textbf{Step 1: Polynomial First Integral}.
With $u=(\Phi')^{2}$ the chain rule gives $\Phi''=\tfrac12\,\dd u/\dd\Phi$,
which turns Eq.~\eqref{eq:reduced-ode-AB} into an equation linear in
$\dd u/\dd\Phi$; its integral is the polynomial first integral
\begin{equation}\label{eq:first-integral-AB}
(\Phi')^{2}=u(\Phi)=a+A\,\Phi^{2}-\frac{B}{2}\,\Phi^{4}+\frac{B^{2}}{3}\,\Phi^{6},
\end{equation}
with $a$ the single integration constant.

\noindent\textbf{Step 2: Quartic in \texorpdfstring{$\Psi=\Phi^{2}$}{Psi=Phi\textasciicircum2}}.
Introduce the squared profile
\begin{equation}
\Psi=\Phi^{2},
\qquad
u=(\Phi')^{2}=\frac{1}{4\Psi}\left(\frac{\dd\Psi}{\dd y}\right)^{2}.
\label{eq:Psi-substitution}
\end{equation}
Substituting Eq.~\eqref{eq:first-integral-AB} into Eq.~\eqref{eq:Psi-substitution}
gives
\begin{equation}\label{eq:elliptic-eq}
\left(\frac{\dd\Psi}{\dd y}\right)^{2}=4\Psi\,g(\Psi)=:P_{4}(\Psi),
\qquad
g(\Psi)\equiv a+A\,\Psi-\frac{B}{2}\,\Psi^{2}+\frac{B^{2}}{3}\,\Psi^{3},
\end{equation}
a quartic, hence an elliptic orbit. Its four roots are
real and ordered as
\begin{equation}\label{eq:P-factorization}
P_{4}(\Psi)=\frac{4B^{2}}{3}\,(\Psi-\Psi_{1})(\Psi-\Psi_{2})(\Psi-\Psi_{3})(\Psi-\Psi_{4}),
\qquad
0=\Psi_{1}<\Psi_{2}<\Psi_{3}<\Psi_{4},
\end{equation}
in which the leading coefficient $\tfrac{4B^{2}}{3}$ is positive because
$B>0$. Hence $P_{4}\ge0$ on $[\Psi_{2},\Psi_{3}]$, and the bounded orbit
occupies $\Psi_{2}\le\Psi\le\Psi_{3}$: $\Phi^{2}$ oscillates between the squared
background $\Psi_{2}$ and the squared peak $\Psi_{3}$, at which $\Phi'=0$.

\noindent\textbf{Step 3: Coefficient Matching and Vieta Relations}.
Expanding Eq.~\eqref{eq:P-factorization} and matching the coefficients of
$\Psi^{3}$, $\Psi^{2}$ and $\Psi$ with $4\Psi\,g(\Psi)$ of
Eq.~\eqref{eq:elliptic-eq} gives
\begin{equation}\label{eq:vieta}
\Psi_{2}+\Psi_{3}+\Psi_{4}=\frac{3}{2B},
\qquad
\Psi_{2}\Psi_{3}+\Psi_{2}\Psi_{4}+\Psi_{3}\Psi_{4}=\frac{3A}{B^{2}},
\qquad
a=-\frac{B^{2}}{3}\,\Psi_{2}\Psi_{3}\Psi_{4}.
\end{equation}

\noindent\textbf{Step 4: Existence Domain}.
The stationary points of $g$ are
\begin{equation}\label{eq:g-stationary}
\Psi_{\mp}=\frac{1\mp\sqrt{1-4A}}{2B},
\qquad A<\frac{1}{4},
\end{equation}
with $\Psi_{-}$ the local maximum and $\Psi_{+}$ the
local minimum. Three
positive roots of $g$, and hence the ordering Eq.~\eqref{eq:P-factorization},
require $g(\Psi_{-})\ge0$ and $g(\Psi_{+})\le0$ together with $g(0)=a\le0$. For the periodic and localised family considered here we restrict to $0<A<3/16$, the upper bound being derived in Step~6; in this window these conditions give
\begin{equation}\label{eq:a-critical}
a\in[a_{\mathrm{crit}},\,0],
\qquad
a_{\mathrm{crit}}=-\left[A\,\Psi_{-}-\frac{B}{2}\,\Psi_{-}^{2}+\frac{B^{2}}{3}\,\Psi_{-}^{3}\right].
\end{equation}
At $a=a_{\mathrm{crit}}$ the roots $\Psi_{2}$ and $\Psi_{3}$ coincide and the
band collapses to the constant-amplitude profile; at $a=0$ the background
$\Psi_{2}$ vanishes and the profile is the localised one of Step~6.

\noindent\textbf{Step 5: Second Integration and Closed Profile}.
The quartic Eq.~\eqref{eq:elliptic-eq} is reduced to a cubic by the inversion
$v=1/\Psi$. Since $\Psi'=-v'/v^{2}$ and
$\tfrac{4B^{2}}{3}\Psi_{2}\Psi_{3}\Psi_{4}=-4a$ by Eq.~\eqref{eq:vieta},
\begin{equation}\label{eq:inversion}
\left(\frac{\dd v}{\dd y}\right)^{2}
=(-4a)\left(\frac{1}{\Psi_{2}}-v\right)
\left(v-\frac{1}{\Psi_{3}}\right)
\left(v-\frac{1}{\Psi_{4}}\right),
\qquad -4a>0,
\end{equation}
a cubic with positive leading coefficient whose three roots are ordered as
$\tfrac{1}{\Psi_{4}}<\tfrac{1}{\Psi_{3}}<\tfrac{1}{\Psi_{2}}$, the orbit lying
between the two larger ones. Separating variables in Eq.~\eqref{eq:inversion},
\begin{equation}\label{eq:elliptic-integral}
\int^{v}\frac{\dd w}{
\sqrt{\bigl(\tfrac{1}{\Psi_{2}}-w\bigr)
\bigl(w-\tfrac{1}{\Psi_{3}}\bigr)
\bigl(w-\tfrac{1}{\Psi_{4}}\bigr)}}
=\pm\sqrt{-4a}\;(y-y_{0}) ,
\end{equation}
in which $y_{0}$ is the constant of integration. The substitution
\begin{equation}\label{eq:z-substitution}
z^{2}=\frac{\tfrac{1}{\Psi_{2}}-v}{\tfrac{1}{\Psi_{2}}-\tfrac{1}{\Psi_{3}}},
\qquad
v=\frac{1}{\Psi_{2}}-\left(\frac{1}{\Psi_{2}}-\frac{1}{\Psi_{3}}\right)z^{2},
\end{equation}
brings Eq.~\eqref{eq:elliptic-integral} to the Legendre form
\begin{equation}\label{eq:legendre}
\kappa\,\dd y=\frac{\dd z}{\sqrt{(1-z^{2})(1-mz^{2})}}
\quad\Longrightarrow\quad
z=\sn\bigl(\kappa\,(y-y_{0}),m\bigr),
\end{equation}
with the integration constant chosen so that $z=0$, hence $v=1/\Psi_{2}$ and
$\Psi=\Psi_{2}$, falls at $y=y_{0}$. The modulus, the frequency and the
characteristic of the third kind are
\begin{equation}\label{eq:parametrization}
m=\frac{\Psi_{4}(\Psi_{3}-\Psi_{2})}{\Psi_{3}(\Psi_{4}-\Psi_{2})},
\qquad
\kappa=\sqrt{-a\,\frac{\Psi_{4}-\Psi_{2}}{\Psi_{2}\Psi_{4}}},
\qquad
\nu=1-\frac{\Psi_{2}}{\Psi_{3}},
\end{equation}
and inverting Eq.~\eqref{eq:z-substitution} gives the closed profile
\begin{equation}\label{eq:cnoidal-solution}
\Psi(y)=\frac{\Psi_{2}}{1-\nu\,\sn^{2}\bigl(\kappa\,(y-y_{0}),m\bigr)},
\qquad
\Phi(y)=\sqrt{\Psi(y)} .
\end{equation}
The period is
\begin{equation}\label{eq:period}
T=\frac{2\,\K(m)}{\kappa}.
\end{equation}

\noindent\textbf{Step 6: Normalisation and Localised Endpoint}.
For the periodic branch the normalisation is imposed per period. Since
Eq.~\eqref{eq:cnoidal-solution} gives $\Phi^{2}=\Psi_{2}/(1-\nu\,\sn^{2})$,
integrating over one period Eq.~\eqref{eq:period} gives
\begin{equation}\label{eq:period-integral}
\int_{0}^{T}\Phi^{2}(y)\,\dd y
=\frac{2\,\Psi_{2}}{\kappa}\int_{0}^{\K(m)}\frac{\dd u}{1-\nu\,\sn^{2}(u,m)},
\end{equation}
the $u$-integral being the complete elliptic integral of the third kind, with
characteristic $\nu$ and modulus $m$. The normalisation
$\int_{0}^{T}\Phi^{2}\,\dd y=1$ applied to Eq.~\eqref{eq:period-integral} is the
single relation
\begin{equation}\label{eq:normalization}
\frac{2\,\Psi_{2}}{\kappa}\int_{0}^{\K(m)}\frac{\dd u}{1-\nu\,\sn^{2}(u,m)}=1
\end{equation}
between $A$, $a$ and $B$.

At $a=0$, the upper end of the domain Eq.~\eqref{eq:a-critical}, the factor $\Psi$
of Eq.~\eqref{eq:elliptic-eq} makes $\Psi_{2}=0$, and $g$
is quadratic. Writing
\begin{equation}\label{eq:soliton-s}
s\equiv\sqrt{1-\frac{16A}{3}},
\end{equation}
its two roots are
\begin{equation}\label{eq:soliton-parameters}
\Psi_{3}=\frac{3(1-s)}{4B},
\qquad
\Psi_{4}=\frac{3(1+s)}{4B},
\qquad
0<A<\frac{3}{16},
\end{equation}
the upper bound being the condition that $s$ be real and
$\Psi_{3}$ positive. At $A=3/16$ one has $s=0$ and the profile becomes constant; this degenerate endpoint is therefore excluded from the localised family.
The profile Eq.~\eqref{eq:cnoidal-solution} then becomes the localised one
\begin{equation}\label{eq:bright-soliton}
\Phi^{2}(y)=\frac{4A}{B\left[1+s\cosh\bigl(2\sqrt{A}\,(y-y_{s})\bigr)\right]},
\qquad
\Phi(\pm\infty)=0,
\end{equation}
whose peak is $\Phi^{2}(y_{s})=\Psi_{3}=\tfrac{3(1-s)}{4B}$. Its intensity is
\begin{equation}\label{eq:soliton-intensity}
\int_{-\infty}^{\infty}\Phi^{2}(y)\,\dd y
=\frac{\sqrt{3}}{2B}\,\ln\frac{3-16A}{\bigl(\sqrt{3}-4\sqrt{A}\bigr)^{2}}.
\end{equation}
Setting Eq.~\eqref{eq:soliton-intensity} equal to one solves in closed form.
By Eq.~\eqref{eq:soliton-s} both $3-16A$ and $\sqrt{3}-4\sqrt{A}$ reduce to
expressions in $s$ alone, so the logarithm collapses and the normalisation
becomes an inversion of the hyperbolic tangent,
\begin{equation}\label{eq:soliton-closed-form}
\frac{3-16A}{\bigl(\sqrt{3}-4\sqrt{A}\bigr)^{2}}
=\left(\frac{1+\sqrt{1-s^{2}}}{s}\right)^{2}
\quad\Longrightarrow\quad
\mathrm{artanh}\sqrt{1-s^{2}}=\frac{B}{\sqrt{3}} .
\end{equation}
Inverting the second member and restoring $A$ from Eq.~\eqref{eq:soliton-s}
gives Eq.~\eqref{eq:parameter-3}. As $B$ runs over $(0,\infty)$ the hyperbolic
tangent maps $B/\sqrt{3}$ one to one onto $(0,1)$, so $A$ runs over
$\bigl(0,\tfrac{3}{16}\bigr)$ and the normalisation is a bijection between the
two intervals. In the small-argument limit the same inversion reduces to
\begin{equation}\label{eq:soliton-normalisation}
A\approx\frac{B^{2}}{16}=\frac{J^{2}}{64\pi^{2}}
\qquad (A\ll1).
\end{equation}

The endpoint is approached continuously: as $a\to0^{-}$, $g(\Psi_{2})=0$ gives
$\Psi_{2}=-a/A+O(a^{2})$, hence $\min_{y}\Phi\sim\sqrt{-a/A}$, $\nu,m\to1$ and
$T\to\infty$ by Eq.~\eqref{eq:period}, and Eq.~\eqref{eq:cnoidal-solution}
converges pointwise to Eq.~\eqref{eq:bright-soliton}.

\section{Proof of Theorem~\ref{thm:reciprocal-density}}
\label{app:proof-reciprocal}

\begin{proof}
\smallskip
\noindent\textbf{Step 1: On-shell momentum in the amplitude sector.}
On $\mathcal{T}$ the branch-summed action is
\begin{equation}
\mathcal{S}(y,\Phi,t)=-2\omega t-2\ln\Phi(y)+\ln\tfrac{J}{2\pi}.
\label{eq:S-on-T}
\end{equation}
Its derivative with
respect to the amplitude coordinate is $-2/\Phi(y)$, and the amplitude momentum is
\begin{equation}
\Pi_\Phi\equiv\frac{1}{2}\frac{\partial\mathcal{S}}{\partial\Phi(y)}=-\frac{1}{\Phi(y)}.
\label{eq:PiPhi-onT}
\end{equation}
Substitution of Eq.~\eqref{eq:PiPhi-onT} into the contact potential gives the profile equation, Eq.~\eqref{eq:reduced-ode}, with the first integral Eq.~\eqref{eq:first-integral} determining $\Phi$.

\smallskip
\noindent\textbf{Step 2: Canonical lift to the field momentum.}
The type-2 generating function Eq.~\eqref{eq:generating_functional} gives the
momentum transformation Eq.~\eqref{eq:Pi_to_PiPhi}. Evaluated on $\mathcal{T}$
($\theta=-\omega t$),
\begin{equation}
\Pi_e^*(x,y,t)=\sqrt{\tfrac{2\pi}{J}}\,
e^{-i(k_xx-\omega t)}\,\Pi_\Phi(y)
=-\sqrt{\tfrac{2\pi}{J}}\,\frac{e^{-i(k_xx-\omega t)}}{\Phi(y)}.
\label{eq:Psi-star-lift}
\end{equation}
Multiplying by the field representative $W_e=\sqrt{J/2\pi}\,e^{i(k_xx-\omega t)}\Phi(y)$ of Eq.~\eqref{eq:W_to_Phi} gives, pointwise on $\mathcal{T}$,
\begin{equation}
\Pi_e^*(x,y,t)\,W_e(x,y,t)
=-\sqrt{\tfrac{2\pi}{J}}\frac{e^{-i(\cdots)}}{\Phi}
\cdot\sqrt{\tfrac{J}{2\pi}}e^{i(\cdots)}\Phi
=-1,
\label{eq:Psi-star-W-minus-one}
\end{equation}
i.e. the contact Legendre relation is radial:
\begin{equation}
\Pi_e^*(\mathbf r,t)=-\frac{1}{W_e(\mathbf r,t)}.
\label{eq:Legendrian-radial-full}
\end{equation}

\smallskip
\noindent\textbf{Step 3: Radial derivative of the action.}
Write the reduced field in polar form $W_e=|W_e|e^{i\vartheta}$, so that
$\partial W_e/\partial |W_e|=W_e/|W_e|$ and, by hermiticity
Eq.~\eqref{eq:complex-derivative-convention},
$\partial \mathcal{S}/\partial W_e^*=\overline{\partial \mathcal{S}/\partial W_e}=\overline{\Pi_e^*}$.
The chain rule for the radial derivative of the action is
\begin{equation}
\frac{\partial \mathcal{S}}{\partial |W_e|}
=\frac{\partial \mathcal{S}}{\partial W_e}\frac{\partial W_e}{\partial |W_e|}
+\frac{\partial \mathcal{S}}{\partial W_e^*}\frac{\partial W_e^*}{\partial |W_e|}
=\Pi_e^*\frac{W_e}{|W_e|}+\overline{\Pi_e^*}\frac{W_e^*}{|W_e|}.
\label{eq:chain-rule-radial}
\end{equation}
Substitute the on-shell relation Eq.~\eqref{eq:Legendrian-radial-full},
$\Pi_e^*=-1/W_e$, so that $\overline{\Pi_e^*}=-1/W_e^*$:
\begin{equation}
\frac{\partial \mathcal{S}}{\partial |W_e|}
=-\frac{1}{W_e}\frac{W_e}{|W_e|}-\frac{1}{W_e^*}\frac{W_e^*}{|W_e|}
=-\frac{1}{|W_e|}-\frac{1}{|W_e|}
=-\frac{2}{|W_e|}.
\label{eq:radial-derivative-two-branch}
\end{equation}
Integrating
Eq.~\eqref{eq:radial-derivative-two-branch} at fixed $t,y$, with the integration constant fixed by a reference amplitude $|W_e|_s>0$, gives the real-valued on-shell action required by Theorem~\ref{thm:measure},
\begin{equation}
\mathcal{S}(W_e,t,y)=-2\ln\frac{|W_e|}{|W_e|_s}.
\label{eq:S-two-branch}
\end{equation}

\smallskip
\noindent\textbf{Step 4: Exponentiation gives the density.}
Theorem~\ref{thm:measure} gives the contact probability density, with $n=\tfrac12\operatorname{rank}\dd\Theta_{\mathbb{C}}\vert_{\ker\Theta_{\mathbb{C}}}=1$ for the one-dimensional fibre. Exponentiate Eq.~\eqref{eq:S-two-branch}:
\begin{equation}
\mathcal{P}[W]=\frac{1}{\mathcal{Z}}
\left(\frac{|W_e|}{|W_e|_s}\right)^{-2}
=\frac{|W_e|_{s}^{2}}{\mathcal{Z}}\,\frac{1}{|W_e|^{2}}.
\label{eq:pdf-exponentiated}
\end{equation}
The reference amplitude is fixed for the radial integration at fixed $t,y$, so the amplitude dependence is
\begin{equation}
\mathcal{P}[W]\propto\frac{1}{|W_e|^{2}},
\label{eq:reciprocal-density-proof}
\end{equation}
which is Eq.~\eqref{eq:reciprocal-density}.
\end{proof}

\section{Derivations for Section~\ref{sec:Phase_Diagram_m}}
\label{app:transition}

Throughout, $A$ is held fixed and $\nu$ is the single parameter of the family; $x$, $y$, $m$, $\kappa$, $B$ and $J$ are the closed forms Eqs.~\eqref{eq:xyz-def}--\eqref{eq:J_of_m_exact}, and $r=\sqrt{1-4A}$.

\noindent\textbf{Step 1: Intensity $J(\nu)$}.
\label{app:E1}
By Eq.~\eqref{eq:J_of_m_exact}, $J=2\pi B$ and $B$ is a product of four factors of $\nu$, so
\begin{equation}
\frac{J'}{J}=\frac{x'}{x}-\frac{1}{1-\nu}
+\frac{\dd}{\dd\nu}\ln\!\int_{0}^{\K(m)}\frac{\dd u}{1-\nu\,\sn^{2}(u,m)}
-\frac{\kappa'}{\kappa},
\label{eq:J-logderiv}
\end{equation}
in which the last integral depends on $\nu$ both explicitly and through $m(\nu)$. At $\nu=0$ each of the four logarithmic derivatives is a closed form in $A$. Writing $x_{0}=x(A,0)=\tfrac12(1-r)$, implicit differentiation of the quadratic satisfied by $x$ gives $x'(0)=x_{0}/2$, hence
\begin{equation}
\frac{x'(0)}{x_{0}}=\frac12,\qquad
m'(0)=\frac{\tfrac32-2x_{0}}{\tfrac32-3x_{0}}=\frac13\left(2+\frac1r\right),\qquad
\frac{\kappa'(0)}{\kappa(0)}=\frac14+\frac{x_{0}}{6-12x_{0}}=\frac1{12}\left(2+\frac1r\right).
\label{eq:E1-pieces}
\end{equation}
For the integral, its value at $(\nu,m)=(0,0)$ is $\K(0)=\tfrac\pi2$; its explicit $\nu$-derivative there is $\int_{0}^{\pi/2}\sin^{2}u\,\dd u=\tfrac\pi4$, contributing $\tfrac12$ to the logarithmic derivative, and its $m$-derivative there is $\K'(0)=\tfrac\pi8$, contributing $\tfrac14$, which the chain rule multiplies by $m'(0)$:
\begin{equation}
\frac{\dd}{\dd\nu}\ln\!\int_{0}^{\K(m)}\frac{\dd u}{1-\nu\,\sn^{2}(u,m)}\bigg|_{\nu=0}
=\frac12+\frac{m'(0)}{4}=\frac1{12}\left(8+\frac1r\right).
\label{eq:E1-Ipiece}
\end{equation}
Substituting Eq.~\eqref{eq:E1-pieces} and Eq.~\eqref{eq:E1-Ipiece} into Eq.~\eqref{eq:J-logderiv} at $\nu=0$,
\begin{equation}
\frac{J'(0)}{J(0)}=\frac12-1+\frac1{12}\left(8+\frac1r\right)-\frac1{12}\left(2+\frac1r\right)=0,
\label{eq:E1-Jprime0}
\end{equation}
identically in $A$.

At the extended end, $J'(0)=0$ by Eq.~\eqref{eq:E1-Jprime0} and $J''(0)<0$ by Eq.~\eqref{eq:curvatures} at the $A=0.04$ used throughout, so $J'(\nu)=J''(0)\,\nu+\mathcal{O}\!\left(\nu^{2}\right)<0$ for all sufficiently small $\nu>0$.

\noindent\textbf{Step 2: sn-Logarithm Average}.
\label{app:E2}
The average entering Eq.~\eqref{eq:Seff-nu} is
\begin{equation}
G(\nu,m)=\frac{1}{2\K(m)}\int_{0}^{2\K(m)}\ln\bigl(1-\nu\,\sn^{2}(u,m)\bigr)\,\dd u .
\label{eq:G-def}
\end{equation}
Since $\sn^{2}(2\K-w)=\sn^{2}w$, the integrand is symmetric about the half period and the integral over $[0,2\K]$ is twice that over $[0,\K]$. Substituting $u=\K-w$ in the latter and using $\sn^{2}(\K-w)=\cn^{2}w/\dn^{2}w$ together with $\dn^{2}w=(1-m)+m\cn^{2}w$,
\begin{equation}
\begin{split}
1-\nu\,\sn^{2}(\K-w)&=\frac{(1-m)+ (m-\nu)\,\cn^{2}w}{\dn^{2}w}
=(1-m)\,\frac{1+\tau\,\cn^{2}w}{\dn^{2}w},
\\
\tau&\equiv\frac{m-\nu}{1-m},
\end{split}
\label{eq:G-substitution}
\end{equation}
so that
\begin{equation}
\int_{0}^{\K}\ln\bigl(1-\nu\,\sn^{2}u\bigr)\,\dd u
=\K\ln(1-m)+\int_{0}^{\K}\ln\bigl(1+\tau\,\cn^{2}w\bigr)\,\dd w-2\int_{0}^{\K}\ln\dn w\,\dd w .
\label{eq:G-split-integral}
\end{equation}
The last integral is the dn-logarithm identity. By the quasi-periodicity
\begin{equation}
\dn(z+\K,m)=\frac{\sqrt{1-m}}{\dn(z,m)},
\label{eq:dn-quasiperiodicity}
\end{equation}
the two half-period contributions combine to give
\begin{equation}
\int_{0}^{2\K}\ln\dn z\,\dd z=\frac{\K}{2}\ln(1-m),
\qquad\text{hence}\qquad
\int_{0}^{\K}\ln\dn w\,\dd w=\frac{\K}{4}\ln(1-m),
\label{eq:lndn_identity}
\end{equation}
the two $\ln\dn$ terms in Eq.~\eqref{eq:G-split-integral} then contributing $-\tfrac{\K}{2}\ln(1-m)$. Inserting this into Eq.~\eqref{eq:G-def} yields the split quoted in the text,
\begin{equation}
G(\nu,m)=\frac12\ln(1-m)+\frac{1}{\K(m)}\int_{0}^{\K(m)}\ln\bigl(1+\tau\,\cn^{2}(w,m)\bigr)\,\dd w,
\qquad
\tau=\frac{m-\nu}{1-m}>0,
\label{eq:dn-halfperiod}
\end{equation}
in which the integrand is smooth for every $0<\nu<1$, since $m>\nu$ on the open family. Setting $\nu=0$ makes the integrand of Eq.~\eqref{eq:G-def} equal to $\ln 1$, so $G(0,m)\equiv0$ identically in $m$ and $G'(0)=\partial_{\nu}G(0,0)$. Differentiating Eq.~\eqref{eq:G-def} under the integral sign,
\begin{equation}
\begin{gathered}
\partial_{\nu}G(\nu,m)=-\frac{1}{2\K(m)}\int_{0}^{2\K(m)}\frac{\sn^{2}(u,m)}{1-\nu\,\sn^{2}(u,m)}\,\dd u,
\\[5pt]
\partial_{\nu}G(0,0)=-\bigl\langle\sn^{2}(u,0)\bigr\rangle=-\bigl\langle\sin^{2}u\bigr\rangle=-\frac12 .
\end{gathered}
\label{eq:G-nu-deriv}
\end{equation}

\noindent\textbf{Step 3: Stationarity and Threshold $\lambda_c$}.
\label{app:E3}
The quantity $x$ solves $x^{2}(\nu^{2}-3\nu+3)-\tfrac32(2-\nu)x+3A=0$. Differentiating implicitly and evaluating at $\nu=0$, where $x=x_{0}=\tfrac12(1-r)$ and $\nu^{2}-3\nu+3=3$,
\begin{equation}
\bigl[2x_{0}\cdot3-\tfrac32\cdot2\bigr]x'(0)+\bigl[-\tfrac32 x_{0}\cdot(-1)-3x_{0}^{2}\bigr]=0
\quad\Longrightarrow\quad
x'(0)=\frac{x_{0}}{2},
\qquad
(\ln x)'(0)=\frac12 ,
\label{eq:E3-lnx}
\end{equation}
valid for every $A<\tfrac14$. Together with $(\ln(1-\nu))'\big|_{0}=-1$ and $-\partial_{\nu}G(0,0)=+\tfrac12$ from Eq.~\eqref{eq:G-nu-deriv}, this gives the first cancellation quoted in Eq.~\eqref{eq:Fprime0},
\begin{equation}
S_{\text{eff}}'(0)=\frac12-1+\frac12=0 .
\label{eq:E3-Sprime0}
\end{equation}
The second derivatives follow by carrying the same series to one order higher. Expanding $x$, $y$, $m$, $\kappa$, $G$ and the third-kind integral in powers of $\nu$ with $A$ kept symbolic, and using $G=-\sum_{n\ge1}\frac{\nu^{n}}{n}M_{n}(m)$ with $M_{n}(m)=I_{n}(m)/I_{0}(m)$ and
\begin{equation}
I_{j}(m)=\int_{0}^{\pi/2}\frac{\sin^{2j}\theta}{\sqrt{1-m\sin^{2}\theta}}\,\dd\theta
=\frac{\pi}{2}\sum_{k\ge0}\frac{\binom{2k}{k}}{4^{k}}\,\frac{\binom{2(j+k)}{j+k}}{4^{j+k}}\,m^{k},
\label{eq:E3-moments}
\end{equation}
so that every coefficient is a polynomial with rational coefficients, one obtains the two curvatures of Eq.~\eqref{eq:curvatures}, and hence $\lambda_c=S_{\text{eff}}''(0)/J''(0)$ of Eq.~\eqref{eq:lambda_c}. On $0<A<\tfrac3{16}$, $S_{\text{eff}}''(0)$ is negative throughout, while $5-5r-16rA$ changes sign once; solving $J''(0)=0$ gives a single root in $(0,\tfrac14)$,
\begin{equation}
A_{*}=\frac{2\sqrt{6}-3}{16}.
\label{eq:A-star}
\end{equation}
Since $S_{\text{eff}}''(0)<0$ throughout, $\lambda_c>0$ exactly on $0<A<A_{*}$, and $\lambda_c\to+\infty$ as $A\to A_{*}^{-}$.

\noindent\textbf{Step 4: Order of the Critical Point}.
\label{app:E4}
The order of the critical point follows from an identity of Taylor coefficients. Carrying the series of Appendix~\ref{app:E3} to $\nu^{4}$ and comparing coefficients,
\begin{equation}
\bigl[\nu^{3}\bigr]S_{\text{eff}}-\bigl[\nu^{2}\bigr]S_{\text{eff}}=0,
\qquad
\bigl[\nu^{3}\bigr]J-\bigl[\nu^{2}\bigr]J=0,
\label{eq:E4-identity}
\end{equation}
both vanishing identically in $A$ under symbolic simplification; since $f''(0)=2[\nu^{2}]f$ and $f'''(0)=6[\nu^{3}]f$, both $S_{\text{eff}}$ and $J$ satisfy
\begin{equation}
f'''(0)=3f''(0)
\label{eq:third-order}
\end{equation}
identically in $A$. For the three pieces of $S_{\text{eff}}=\ln x+\ln(1-\nu)-G$ the same differences are
\begin{equation}
\bigl([\nu^{3}]-[\nu^{2}]\bigr)\ln x=-\frac1{12},
\quad
\bigl([\nu^{3}]-[\nu^{2}]\bigr)\ln(1-\nu)=+\frac16,
\quad
\bigl([\nu^{3}]-[\nu^{2}]\bigr)(-G)=-\frac1{12},
\label{eq:E4-three-constants}
\end{equation}
the first and the third being independent of $A$ under symbolic simplification, and summing to zero. Consequently $\mathcal{F}'''(0)=3\bigl[S_{\text{eff}}''(0)-\lambda_cJ''(0)\bigr]=0$ at $\lambda=\lambda_c$, and the expansion of $\mathcal{F}'$ about $\nu=0$ at $\lambda=\lambda_c(1+\varepsilon)$ begins at fourth order,
\begin{equation}
\mathcal{F}'(\nu)=-\lambda_c\,\varepsilon\,J''(0)\,\nu+\frac{Q}{6}\,\nu^{3}+\mathcal{O}\!\left(\nu^{4},\varepsilon\nu^{2}\right),
\label{eq:expand}
\end{equation}
with
\begin{equation}
\begin{split}
Q(A)
&=S_{\text{eff}}''''(0)-\lambda_cJ''''(0)\\
&=-\frac{1}{32}\,\frac{N(r,A)}{(r-1)^{2}(1-4A)^{3/2}(r-1+4A)\bigl(-5+5r+16rA\bigr)},
\end{split}
\label{eq:E4-Q}
\end{equation}
\begin{equation}
\begin{split}
N(r,A)=\;&40(r-1)+(5263-5183r)A+(-44459+34173r)A^{2}\\
&-204(-581+296r)A^{3}+256(-338+55r)A^{4}-11264A^{5}.
\end{split}
\label{eq:E4-N}
\end{equation}
At $A=\tfrac1{25}$, $Q<0$, and the separatrix solving $\mathcal{F}'=0$ satisfies
\begin{equation}
\nu^{*}(\lambda)=c\,\varepsilon^{1/2}+\mathcal{O}(\varepsilon),
\qquad
c=\sqrt{\frac{6\,\lambda_c\,J''(0)}{Q}} ,
\label{eq:onset_linear}
\end{equation}
with $c=2.9484408104$. Substituting $A=(1-r^{2})/4$ factorises the numerator as $(r-1)^{4}r$ times a quintic, so $Q$ has a single zero in $(0,A_{*})$, at $A_{Q}$, where $r_{Q}$ is the unique root in $\bigl(\sqrt{1-4A_{*}},\,1\bigr)$ of
\begin{equation}
176r^{5}+1584r^{4}-1008r^{3}-1256r^{2}+945r-135=0 ;
\label{eq:E4-quintic}
\end{equation}
numerically
\begin{equation}
r_{Q}=0.7865553018,
\qquad
A_{Q}=\frac{1-r_{Q}^{2}}{4}=0.0953326893,
\qquad
\lambda_c(A_{Q})=2.0100855205.
\label{eq:A-Q}
\end{equation}
Since $J''(0)<0$ on $(0,A_{*})$, the product $\lambda_cJ''(0)/Q$ is positive exactly where $Q<0$, that is, on $0<A<A_{Q}$, which is where the coefficient $c$ is real and the square-root law holds.

The same $Q$ governs how $h$ leaves $\lambda_c$. Since Eq.~\eqref{eq:third-order} holds for both $S_{\text{eff}}$ and $J$, their derivatives $S_{\text{eff}}'$ and $J'$ share the Taylor coefficients of $\nu$ and of $\nu^{2}$, and differ first at $\nu^{3}$; the ratio therefore starts at order $\nu^{2}$,
\begin{equation}
h(\nu)=\lambda_c\left[1+\frac{Q}{6\,S_{\text{eff}}''(0)}\,\nu^{2}+\mathcal{O}\!\left(\nu^{3}\right)\right],
\label{eq:E4-h-onset}
\end{equation}
so $h'(0)=0$, and since $S_{\text{eff}}''(0)<0$ throughout, $h$ rises from $\lambda_c$ exactly on $0<A<A_{Q}$ and falls below it on $A_{Q}<A<A_{*}$.

At $A=0.04$, $Q<0$ and $S_{\text{eff}}''(0)<0$, so Eq.~\eqref{eq:E4-h-onset} gives
$h''(0)=\lambda_c Q/[3S_{\text{eff}}''(0)]>0$ and $h(\nu)>\lambda_c$ for
$0<\nu\ll1$. For
$\lambda=\lambda_c(1+\varepsilon)$ with $0<\varepsilon\ll1$, the equation
$h(\nu)=\lambda$ has the solution $\nu^*(\lambda)$ of
Eq.~\eqref{eq:onset_linear} in a neighbourhood of $\nu=0$. At $\nu^*$,
\begin{equation}
\mathcal{F}''(\nu^*)=J'(\nu^*)\,h'(\nu^*)<0,
\label{eq:E4-Fpp-local}
\end{equation}
so $\mathcal{F}$ has an interior maximum and the weight an interior minimum.
This is the bifurcation at $\lambda_c$.

\noindent\textbf{Step 5: Infrared Bound}.
\label{app:E6}
Exponentiating Eq.~\eqref{eq:Seff-nu}, the sector weight is
\begin{equation}
w(\nu)\equiv e^{-\mathcal{F}(\nu,\lambda)}
=\frac{e^{\lambda J(\nu)}}{x(\nu)\,(1-\nu)}\,e^{\,G(\nu,m(\nu))},
\label{eq:E6-weight}
\end{equation}
in which $x$ tends to a finite positive limit and $G$ to $\tfrac12\ln(1-m)$ as $\nu\to1$, while $m\to1$ there. The leading singularity is therefore the explicit factor $(1-\nu)^{-1}$ tempered by $e^{G}\sim(1-\nu)^{1/2}$, that is,
\begin{equation}
w(\nu)\sim(1-\nu)^{-1/2}
\qquad(\nu\to1),
\label{eq:E6-divergence}
\end{equation}
which is integrable, so Eq.~\eqref{eq:Z_m_decomp} converges even at $\nu_{\max}=1$. The bound comes from requiring $|W_e|\ge R_{\min}$ across a transverse period, that is, from bounding the trough value,
\begin{equation}
|W_e|^{2}_{\min}=B\Psi_{2}=x(A,\nu)\,(1-\nu)\;\ge\;R_{\min}^{2}=I_{\min}\sim V^{-1},
\label{eq:E6-IR-condition}
\end{equation}
and since $x\to3(1-s)/4$ with $s=\sqrt{1-16A/3}$ as $\nu\to1$, solving Eq.~\eqref{eq:E6-IR-condition} to leading order in $R_{\min}^{2}$ gives the cap $\nu_{\max}$ of Eq.~\eqref{eq:IR-domain}. For $A=0.04$ the coefficient is $4/\bigl[3(1-s)\bigr] \sim 11.79$, so the leading order gives $\nu_{\max} \sim 1$ at $V=10^{6}$.

\section{Stochastic Simulations}
\label{app:stoch-radial}

The radial density $\mathcal{P}[|W|]\propto|W_e|^{-2}$ of Eq.~\eqref{eq:radial-measure-two-branch} is the density of Theorem~\ref{thm:measure} restricted to the on-shell family: Appendix~\ref{app:proof-reciprocal} obtains it from the Legendre relation $\Pi_{e}^{*}W_{e}=-1$, Eq.~\eqref{eq:Psi-star-W-minus-one}, without any stochastic input. On the spatially uniform amplitude direction $W=R\,e^{i\theta}$ it reads $p\propto R^{-2}$ on $[R_{\min},R_{\max}]$, with the cut-offs $R_{\min}$ and $R_{\max}$ of Section~\ref{sec:cont_norm}.

In this sector the radial component of Eq.~\eqref{eq:qcgle} is
\begin{equation}
\dot R=\mu R+R^{3}-R^{5}=R\,g(R^{2}),
\qquad
g(\mathcal{N}_{+})=0,
\label{eq:G-drift}
\end{equation}
where $g$ is the uniform-sector balance of Eq.~\eqref{eq:N-uniform}; the coefficients $c_{1}$ and $c_{3}$ enter only the phase equation. The drift is therefore fixed by the CGLE, and by Eq.~\eqref{eq:G-drift} it vanishes at the saturation cut-off $R_{\max}$ and is positive on the whole interior $(R_{\min},R_{\max})$. Consider radial diffusions driven by this drift,
\begin{equation}
\dd R=\bigl(\mu R+R^{3}-R^{5}\bigr)\,\dd t+\sqrt{2D(R)}\,\dd W_{t},
\label{eq:G-sde}
\end{equation}
where $W_{t}$ is a standard Wiener process and the ends $R_{\min}$, $R_{\max}$ reflect.

For a prescribed stationary density $p$ the zero-flux condition of the It\^o Fokker--Planck equation reads
\begin{equation}
(Dp)'=\bigl(\mu R+R^{3}-R^{5}\bigr)p,
\qquad
D(R_{\min})=0,
\qquad
D=p^{-1}\int_{R_{\min}}^{R}\bigl(\mu s+s^{3}-s^{5}\bigr)\,p(s)\,\dd s,
\label{eq:G-zero-flux}
\end{equation}
so drift and density together fix $D$, and a diffusion exists only where the result is non-negative. For the density of Appendix~\ref{app:proof-reciprocal} this gives
\begin{equation}
\begin{gathered}
D(R)=R^{2}\int_{R_{\min}}^{R}\bigl(\mu s+s^{3}-s^{5}\bigr)s^{-2}\,\dd s
=R^{2}\left[\mu\ln\frac{R}{R_{\min}}+\frac{R^{2}-R_{\min}^{2}}{2}-\frac{R^{4}-R_{\min}^{4}}{4}\right],\\
D(R_{\min})=0,\qquad D>0\ \text{on}\ (R_{\min},R_{\max}],\qquad D(R_{\max})=4.357\ \text{at}\ \mu=1,\ V=100,
\end{gathered}
\label{eq:G-D}
\end{equation}
so the pair is admissible. With $I=R^{2}$ and $\dd I=2R\,\dd R$ the stationary pair of Eq.~\eqref{eq:G-sde} is
\begin{equation}
p(R)=\mathcal{N}\,R^{-2},
\qquad
P(I)=\tfrac{\mathcal{N}}{2}\,I^{-3/2},
\qquad
\mathcal{N}=\Bigl[\tfrac{1}{R_{\min}}-\tfrac{1}{R_{\max}}\Bigr]^{-1},
\label{eq:G-laws}
\end{equation}
in which $D$ does not appear, so the intensity exponent $-3/2$ holds exactly (Eq.~\eqref{eq:IntensityDist}). The zero-flux condition solved above is the steady-state balance of the radial sector between injection by the noise and dissipation by the CGLE drift, the same kind of identity as the energy-balance identity of \cite{Kuksin2004StationaryMeasures}, their Eq.~(1.27); there the identity runs from the force to the stationary measure, here from the measure and the drift to the force, so $D$ is the effective intensity of the fluctuations about the on-shell family that the balance implies.

Replacing $D(R)$ in Eq.~\eqref{eq:G-sde} by a positive constant $D_{0}$ changes the zero-flux density to
\begin{equation}
p(R)\propto\exp\!\Bigl[\bigl(\mu R^{2}/2+R^{4}/4-R^{6}/6\bigr)/D_{0}\Bigr],
\label{eq:G-D0}
\end{equation}
an increasing function of $R$; the same discretisation, $1000$ particles over $T=150$ at $\Delta t=2\times10^{-4}$, returns interior slopes $+0.459$, $+0.229$ and $+0.115$ for $D_{0}=0.5,\,1,\,2$, the sampled densities agreeing with Eq.~\eqref{eq:G-D0} to within $2.1\%$, $1.8\%$ and $1.5\%$.

For the pair of Eq.~\eqref{eq:G-D} the base run uses $\Delta t=2\times10^{-4}$ with $2000$ particles over $T=240$, the first third discarded and every two-hundredth step retained, giving $8.0\times10^{6}$ samples at $\mu=1$, $V=100$. Over the interior of the domain the least-squares slopes of $\log p$ against $\log R$ and of $\log P$ against $\log I$ are $-1.996$ and $-1.498$, $\langle R\rangle=0.2767$ and $\mathrm{std}(R)=0.2265$ against the analytic values $0.2760$ and $0.2259$ of Eq.~\eqref{eq:G-laws} (Fig.~\ref{fig:fig-2}).

The same process, re-run with one cut-off swept at a time, tests the cut-off scaling of the moments predicted by Eq.~\eqref{eq:MomentScaling} (Fig.~\ref{fig:fig-3}). Holding $R_{\min}=0.1$ and sweeping $\mu$ over $[0.5,8]$, which moves the saturation cut-off over $R_{\max}\in[1.169,1.836]$, the moments $\langle I^{p}\rangle$ for $p=1,1.5,2$ follow $R_{\max}^{2p-1}$ to within $1\%$ of the exact radial moment $\mathcal{N}(R_{\max}^{2p-1}-R_{\min}^{2p-1})/(2p-1)$. Holding $\mu=1$ and sweeping $V$, which moves the finite-volume cut-off over $R_{\min}\in[0.05,0.2]$, the moment $\langle I^{1/4}\rangle$ follows $R_{\min}^{1/2}$ to within $0.2\%$, the factor $\mathcal{N}\propto R_{\min}$ cancelling the infrared divergence. A maximum-likelihood fit of the $q$-exponential $e_{q}(-\beta I)$ to the intensity density over the interior returns $\hat q=1.667$ (95\% CI $[1.666,1.668]$), within $0.03\%$ of $5/3$.

The realisation is one-dimensional and radial; it is not the strict reduction of a two-dimensional complex Langevin equation. For isotropic constant noise $D_{0}$ the radial reduction carries the It\^o drift and zero-flux line density
\begin{equation}
b(R)=\mu R+R^{3}-R^{5}+\frac{D_{0}}{R},
\qquad
p\propto R\exp\!\Bigl[\bigl(\mu R^{2}/2+R^{4}/4-R^{6}/6\bigr)/D_{0}\Bigr],
\label{eq:G-2D-pair}
\end{equation}
so matching $R^{-2}$ forces
\begin{equation}
D_{0}=-\tfrac13\,R\bigl(\mu R+R^{3}-R^{5}\bigr),
\label{eq:G-2D-obstruction}
\end{equation}
which is negative on $(R_{\min},R_{\max})$ at every $\mu$, the drift being positive there. Isotropic positive noise with the pure CGLE drift therefore cannot realise the measure in two dimensions; a confining drift $-3D_{0}/R$ restores $p\propto R^{-2}$ exactly.

A second simulation verifies the on-shell bifurcation (Fig.~\ref{fig:fig-5}). A Metropolis chain samples the shape parameter $\nu$ from the equilibrium weight $\exp[-\mathcal{F}(\nu,\lambda)]$ of Eq.~\eqref{eq:cond_prob} in the variable
\begin{equation}
t=-\ln(1-\nu),
\qquad
t_{\max}=-\ln(1-\nu_{\max}),
\qquad
e^{-t}=\frac{\dd\nu}{\dd t},
\label{eq:G-metropolis}
\end{equation}
which maps $\nu\in[0,\nu_{\max})$ onto a finite interval with upper end fixed by the infrared condition Eq.~\eqref{eq:IR-domain}, with Jacobian $e^{-t}$ as in Eq.~\eqref{eq:G-metropolis}. The sample means reproduce the order parameter $\langle\nu\rangle(\lambda)$ of Fig.~\ref{fig:fig-5}, the response is obtained from the same samples through the covariance identity Eq.~\eqref{eq:chi-cov} rather than by differentiating $\langle\nu\rangle$, and the half-maximum point $\lambda_{1/2}$ converges as the volume increases.

\bibliography{bib-Contact-2025_ver2.0.bib} 
\end{document}